\documentclass[journal,10pt]{IEEEtran}
\usepackage{amsmath}
\usepackage{amssymb}
\usepackage{amsfonts}
\usepackage{graphicx}
\usepackage{epsfig}
\usepackage{subfigure}
\usepackage{psfrag}
\usepackage{color}
\usepackage{url}
\usepackage{cite}
\begin{document}
\title{Cooperative Energy Trading in CoMP Systems Powered by Smart Grids}

\author{Jie Xu and Rui Zhang
\thanks{
This paper has been presented in part at IEEE Global Communications Conference (Globecom), Austin, TX USA, December 8-12, 2014.
}
\thanks{J. Xu is with the Department of Electrical and Computer Engineering, National University of Singapore (e-mail: jiexu.ustc@gmail.com).}
\thanks{R. Zhang is with the Department of Electrical and Computer Engineering, National
University of Singapore (e-mail: elezhang@nus.edu.sg). He is also with the Institute for Infocomm Research, A*STAR, Singapore.}}
\maketitle

\begin{abstract}
This paper studies the energy management in the coordinated multi-point (CoMP) systems powered by smart grids, where each base station (BS) with local renewable energy generation is allowed to implement the two-way energy trading with the grid. Due to the uneven renewable energy supply and communication energy demand over distributed BSs as well as the difference in the prices for their buying/selling energy from/to the gird, it is beneficial for the cooperative BSs to jointly manage their energy trading with the grid and energy consumption in CoMP based communication for reducing the total energy cost. Specifically, we consider the downlink transmission in one CoMP cluster by jointly optimizing the BSs' purchased/sold energy units from/to the grid and their cooperative transmit precoding, so as to minimize the total energy cost subject to the given quality of service (QoS) constraints for the users. First, we obtain the optimal solution to this problem by developing an algorithm based on techniques from convex optimization and the uplink-downlink duality. Next, we propose a sub-optimal solution of lower complexity than the optimal solution, where zero-forcing (ZF) based precoding is implemented at the BSs. Finally, through extensive simulations, we show the performance gain achieved by our proposed joint energy trading and communication cooperation schemes in terms of energy cost reduction, as compared to conventional schemes that separately design communication cooperation and energy trading.
\end{abstract}


\IEEEpeerreviewmaketitle
\setlength{\baselineskip}{1\baselineskip}
\newtheorem{definition}{\underline{Definition}}[section]
\newtheorem{fact}{Fact}
\newtheorem{assumption}{Assumption}
\newtheorem{theorem}{\underline{Theorem}}[section]
\newtheorem{lemma}{\underline{Lemma}}[section]
\newtheorem{corollary}{\underline{Corollary}}[section]
\newtheorem{proposition}{\underline{Proposition}}[section]
\newtheorem{example}{\underline{Example}}[section]
\newtheorem{remark}{\underline{Remark}}[section]
\newtheorem{algorithm}{\underline{Algorithm}}[section]
\newcommand{\mv}[1]{\mbox{\boldmath{$ #1 $}}}
\section{Introduction}\label{sec:Introduction}

Due to the explosive increase of mobile data traffic, cellular operators have been deploying denser base stations (BSs) with increasing frequency reuse factor for providing higher capacity to subscribers. However, this gives rise to the more severe inter-cell interference (ICI), which becomes one key issue to be tackled in future wireless networks. To meet this challenge, coordinated multi-point (CoMP) transmission has emerged as one promising technique (see \cite{Gesbert2010,Yang2013How} and the references therein), where multiple BSs cooperatively serve a group of mobile users by implementing baseband signal coordination to transform the harmful ICI into useful information signals for coherent transmission/reception in the downlink and uplink transmissions, respectively.

Another potential challenge faced by cellular operators is their drastically increasing operational costs due to the on-grid energy consumption by the growing number of BSs. Among assorted solutions that are proposed to overcome this issue, equipping BSs with energy harvesters that can harvest energy from the environmental sources, e.g., solar and wind, is a promising solution, since the cost of renewable energy generation is generally lower than that of the conventional energy from the grid \cite{huawei,HanAnsari2014}. Furthermore, with the advancement of smart girds technology, two-way energy and information flows become feasible between distributed loads and the grid for enabling more energy-efficient power networks \cite{XueSmartGird}. As a specific type of energy loads, BSs in cellular networks can thereby be enabled to implement the two-way energy trading with the grid \cite{Chen2013,WangSaadHan,Leithon2013} to more efficiently utilize their locally generated renewable energy for saving the energy cost. With two-way energy trading, BSs of renewable energy surplus can sell its excessive energy to the grid to make profit, while BSs of renewable energy deficit can buy additional energy from the grid to maintain its reliable operation and communication.

In this paper, we pursue a unified study of CoMP based communication cooperation and two-way energy trading enabled energy cooperation in a cellular system powered by smart girds, where each BS is equipped with one or more energy harvesting devices (wind-turbines and/or solar panels) to generate renewable energy locally, and implements the two-way energy trading with the grid. In practice, the renewable generation rates and communication energy demands are both uneven over distributed BSs, which is due to the fact that different BSs may use different types of energy harvesting devices (with distinct energy generation capacities), and face diverse wireless service requests from randomly arrived cellular users. Furthermore, due to different energy supply and demand conditions in the whole grid, the prices for each BS to buy and sell one unit of energy from/to the grid are in general different \cite{Chen2013}. By considering all these practical issues, new research challenges are imposed on the cost-efficient management of energy consumption in CoMP based communication and energy trading with the grid, for which the conventional CoMP designs (see, e.g., \cite{WieselEldarShamai2006,Bengtsson1999} and the references  therein) that ignore the new feature of two-way energy trading in smart grids are no more effective.

\begin{table*}[ht]
\renewcommand{\arraystretch}{1.3}
\caption{Performance Comparison for the Toy Example in Fig. \ref{fig:toy_example}}
\label{table1} \centering
\begin{tabular}{|p{1.7in}|p{0.8in}|p{0.8in}|p{0.8in}|p{0.8in}|p{0.8in}|}
\hline
& Transmit power at BS 1& Transmit power at BS 2& Total energy consumption&Total energy cost\\ \hline
Conventional CoMP design \cite{Gesbert2010}& $\bar P_{{\rm t},1} = 0.64$& $\bar P_{{\rm t},2}= 0.16$ & $\bar Q = 0.8$&$\bar C = 0.356$\\ \hline
Proposed CoMP design with joint energy trading consideration& $\hat  P_{{\rm t},1} = 0.25$ & $\hat  P_{{\rm t},2} = 1$ & $\hat Q = 1.25$&$\hat C = 0.05$
\\ \hline
\end{tabular}
\end{table*}
\begin{figure}
\centering
 \epsfxsize=1\linewidth
    \includegraphics[width=7cm]{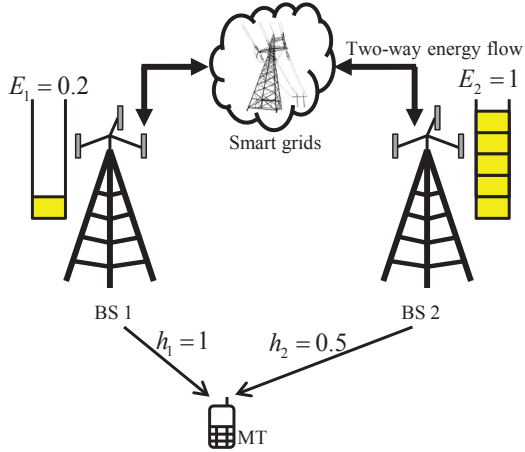}
\caption{A toy example of cellular system with two single-antenna BSs cooperatively transmitting to one single-antenna MT.} \label{fig:toy_example}\vspace{0em}
\end{figure}

Specifically, to illustrate this new challenge, we consider a toy example as depicted in Fig. \ref{fig:toy_example}, where two single-antenna BSs cooperatively transmit to a single-antenna mobile terminal (MT), and the MT has a given quality of service (QoS) requirement to be met. Suppose that the channel coefficients from BS 1 and BS 2 to the MT are given by $h_{1} = 1$ and $h_{2} = 0.5$, respectively, and the background Gaussian noise at the MT has a normalized power of $1$. By assuming that received signals are coherently combined, the resulting signal-to-noise-ratio (SNR) at the MT is then expressed as $\mathtt{SNR} = \left(h_{1}\sqrt{P_{{\rm t},1}} + h_{2}\sqrt{P_{{\rm t},2}}\right)^2$ with $P_{{\rm t},1} \ge 0$ and $P_{{\rm t},2} \ge 0$ denoting the transmit power at BS 1 and BS 2, respectively. We assume the QoS requirement at the MT as a minimum SNR equal to $1$, i.e., $\mathtt{SNR} \ge 1$. Let the energy harvesting rates at the two BSs be denoted by $E_1 = 0.2$ and $E_2 = 1$, and the energy prices for both the two BSs to buy and sell one unit of energy from and to the grid by $\alpha_{\rm b} = 1$ and $\alpha_{\rm s} = 0.1$, respectively. Note that all units are normalized for simplicity. Under the above setup, we compare the conventional CoMP design to minimize the total transmit power of the two BSs (but ignoring the different renewable energy rates at the two BSs as well as the two-way energy trading prices) versus our proposed new CoMP design with joint energy trading consideration (details will be given later), including their respective power, total energy consumption, and total energy cost in Table \ref{table1}. More details on their performance comparison are given as follows.

\begin{itemize}
  \item {\bf Conventional CoMP design \cite{Gesbert2010}:} The two BSs first determine their transmit power values with CoMP downlink transmission to minimize the sum-power to achieve $\mathtt{SNR} = 1$ at the MT. Based on the optimal maximal ratio combining (MRC) principle, the transmit powers at BS 1 and BS 2 are obtained as $\bar P_{{\rm t},1} = 0.64$ and $\bar P_{{\rm t},2}= 0.16$. The total energy consumption of the two BSs is thus $\bar{Q}=\bar P_{{\rm t},1}+\bar P_{{\rm t},2}=0.8$. Next, given the available renewable energy at the two BSs, BS 1 needs to purchase $\bar P_{{\rm t},1} - E_1 = 0.44$ unit of energy from the grid at the price $\alpha_{\rm b}=1$ per unit, while BS 2 should sell $E_2- \bar P_{{\rm t},2} = 0.84$ unit of energy to the grid at the price $\alpha_{\rm s}=0.1$ per unit,  resulting in a total energy cost of $\bar C = \alpha_{\rm b} (\bar P_{{\rm t},1} - E_1) - \alpha_{\rm s} (E_2- \bar P_{{\rm t},2}) = 0.356$.
  \item {\bf Proposed CoMP design with joint energy trading consideration:}  The two BSs jointly optimize the CoMP based communication cooperation and the energy trading with the grid for minimizing the total energy cost (see Section \ref{sec:optimal} for the detailed algorithm). Accordingly, the transmit powers at BS 1 and BS 2 are given by $\hat  P_{{\rm t},1} = 0.25$ and $\hat  P_{{\rm t},2} = 1$. As a result, the total energy consumed is $\hat{Q}=\hat P_{{\rm t},1}+\hat P_{{\rm t},2}=1.25$, and the total energy cost is $\hat C = \alpha_{\rm b} (\hat P_{{\rm t},1} - E_1) - \alpha_{\rm s} (E_2- \hat P_{{\rm t},2}) = 0.05$.
\end{itemize}
It is observed from the above comparison that although the  conventional CoMP design achieves the lowest sum-power consumption (i.e., $\bar{Q} = 0.8$ as compared to $\hat{Q} = 1.25$), it incurs much larger total energy cost than the proposed design (i.e., $\bar C = 0.356$ versus $\hat C = 0.05$). This is because the conventional design optimizes transmit powers without considering the differences in renewable generation rates at the two BSs as well as the energy buying/selling prices with the grid, while the proposed design exploits such differences, so that the cheaper renewable energy is more efficiently utilized and the more expensive energy from the grid is minimized. This toy example, albeit being simplistic, suggests that to reduce the energy cost of CoMP systems powered by smart grids, it is crucial to jointly optimize the BSs' energy management in CoMP transmission and energy trading with the grid, by taking into account the different renewable generation rates over the BSs and the distinct energy buying/selling prices. To our best knowledge, such a joint optimization approach has not been studied in the literature, which motivates this work.

For the purpose of exposition, we consider in this paper the downlink transmission in one CoMP cluster, where a group of multiple-antenna BSs each with local renewable energy generation cooperatively transmit to a set of single-antenna MTs by applying linear transmit precoding. We jointly optimize the BSs' purchased/sold energy units from/to the grid and their cooperative transmit precoding, so as to minimize the total energy cost of the BSs subject to the given QoS constraints for the MTs. First, we obtain the optimal solution to this problem by developing an algorithm based on techniques from convex optimization \cite{BoydVandenberghe2004} and the uplink-downlink duality \cite{YuLan2007}. Next, we propose a sub-optimal solution of lower complexity than the optimal solution, where zero-forcing (ZF) based cooperative transmit precoding \cite{WieselZF,ZhangRuiBD} is considered. Finally, we show by extensive simulations the promising performance gain achieved by our proposed optimal and sub-optimal schemes in terms of energy cost reduction, as compared to conventional schemes with separate communication cooperation and energy trading designs.

It is worth noting that recently there have been several studies on improving the energy efficiency of cellular networks by taking advantage of various smart grid features \cite{BuYu2012,Chia2013,GuoTCOM,XuGuoZhang,XuDuanZhang}. For instance, the utilities of both the cellular network and the power network are jointly optimized in \cite{BuYu2012}. An alternative form of energy cooperation, where distributed BSs exchange their locally harvested energy via dedicated power lines \cite{Chia2013} or utilizing the smart grid infrastructure \cite{GuoTCOM,XuGuoZhang,XuDuanZhang}, is also studied.

It is also worth pointing out that our proposed joint energy trading and communication cooperation approach is more general than that in the prior works on exploiting the two-way energy trading in smart grids only \cite{Chen2013,WangSaadHan,Leithon2013}. The prior studies in \cite{Chen2013,WangSaadHan,Leithon2013} focus on the energy trading management for distributed loads with the grid, by assuming their energy demands to be given. In contrast, in this paper we consider a specific type of energy loads (i.e., BSs in cellular networks for communication) and jointly optimize their energy demands for communication and the two-way energy trading with the grid for reducing the overall energy cost. Our proposed approach thus provides new useful insights on the joint demand- and supply-side energy management in smart grids by considering controllable energy loads through communications design and scheduling.


{\it Notation:} Boldface letters refer to vectors (lower  case) or matrices (upper case). For a square matrix $\mv{S}$, $\mv{S}^{-1}$  denotes its inverse, while $\mv{S}\succ \mv{0}$ means that $\mv{S}$ is positive definite. For an arbitrary-size matrix $\mv{M}$, $\mv{M}^H$ and $\mv{M}^T$ denote the conjugate transpose and transpose of $\mv{M}$, respectively, and $[{\mv{M}}]_{kl}$ denotes the element in the $k$th row and $l$th column of ${\mv{M}}$. $\mv{I}$ and $\mv{0}$ denote an identity matrix and an all-zero matrix, respectively, with appropriate dimensions. $\|\mv{x}\|$ denotes the Euclidean norm of a complex vector $\mv{x}$, and $|z|$ denotes the magnitude of a complex number $z$. For a real number $z$, $(z)^+ \triangleq\max(0,z)$. Symbol $j$ denotes the complex number $\sqrt{-1}$.

\section{System Model and Problem Formulation}\label{sec:system}

We consider practical cluster-based CoMP systems by focusing our study on one single cluster as shown in Fig. \ref{fig:1}, in which $N>1$ BSs each equipped with $M\ge 1$ antennas cooperatively send independent messages to $K$ single-antenna MTs. For convenience, we denote the set of MTs and that of BSs as $\mathcal{K} = \{1,\ldots,K\}$ and $\mathcal{N}=\{1,\ldots,N\}$, respectively. We assume that each BS is locally equipped with one or more energy harvesting devices (wind turbines and/or solar panels), and is also connected to the smart grid for implementing the two-way energy trading. We also assume that there is a central unit deployed for each CoMP cluster which coordinates the cooperative energy trading as well as the cooperative communication within the cluster. To this end, the central unit needs to collect both the communication data (i.e., the transmit messages and channel state information (CSI)) from each of the BSs through the cellular backhaul links, and the energy information (i.e., the energy harvesting rates and energy buying/selling prices) via the smart meters installed at BSs and the grid-deployed communication/control links connecting them. The central unit can be one of the $N$ BSs that serves as the cluster head, or a dedicatedly deployed entity in the network.

\begin{figure}
\centering
 \epsfxsize=1\linewidth
    \includegraphics[width=8cm]{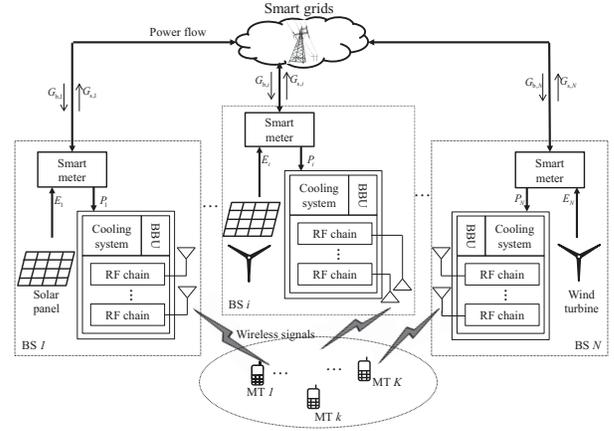}
\caption{A clustered CoMP system powered by smart grids, where the BSs have local renewable energy generation and implement the two-way energy trading with the grid.} \label{fig:1}
\end{figure}

We assume block-based transmissions and quasi-static models for both the renewable energy processes and wireless channels, where the energy harvesting rates and the channel coefficients remain constant during each communication block and may change from one block to another. This model is practical, since the coherence time of a wireless channel (say, several milliseconds) is usually much smaller than that of an energy harvesting process (e.g., a few tens of seconds for solar and wind power). For convenience, each block duration is normalized to unity unless otherwise specified; thus, the terms ``energy'' and ``power'' will be used interchangeably in the sequel. For the purpose of investigation, in this paper we do not consider energy storage devices used at each of the BSs due to their high deployment and maintenance costs, and will leave the case with finite energy storage at BSs for our future work. As a result, each BS will either consume all of its harvested energy for communication or sell any excessive energy to the grid during each block. This simplifies our analysis to one particular block for investigation, and also helps provide insights on the achievable gain by cooperative energy trading and communication cooperation. Next, we explain the energy management model at the BSs, then present the downlink CoMP transmission model, and finally formulate the optimization problem for joint energy trading and communication cooperation.

\subsection{Energy Management Model}

As assumed above, each BS is equipped with energy harvesting devices and is also connected to the grid for two-way energy trading. We denote the harvested energy at each BS $i\in\mathcal{N}$ as $E_i \ge 0$, which is a given constant for one block of our interest. We also denote the energy purchased (sold) from (to) the grid at BS $i$ as $G_{{\rm{b}},i} \ge 0$ ($G_{{\rm{s}},i}\ge 0$).{\footnote{Note that in practice, the purchased (sold) energy $G_{{\rm{b}},i}$ ($G_{{\rm{s}},i}$) by each BS $i$ from (to) the grid needs to be upper-bounded due to the limited capacity of the power lines connecting them. This is ensured here since $G_{{\rm{b}},i}$ and $G_{{\rm{s}},i}$ are no larger than the maximum power consumption at BS $i$ (to be specified later) and the maximum harvested energy of $E_i$, respectively, both of which are set to be smaller than the capacity of the power lines. Also note that cellular BSs are just one among many types of power loads in the smart grid, and thus the total energy sold by all BSs to the grid (i.e., $\sum_{i\in\mathcal N}G_{{\rm{s}},i}$) can be temporally larger or smaller than that purchased from the grid (i.e., $\sum_{i\in\mathcal N}G_{{\rm{b}},i}$). For example, in the case of more energy sold to the gird, the smart grid can schedule the excessive energy to supply other conventional loads (without renewable energy supply) in the grid.}} When each BS $i$ buys (sells) one unit energy from (to) the grid, we denote the price that it needs to pay to (or will be paid by) the grid as $\alpha_{{\rm{b}},i} > 0$ ($\alpha_{{\rm{s}},i} > 0$). Then we have the net energy cost at BS $i$ as
\begin{align}\label{total:cost}
C_i = \alpha_{{\rm b},i}G_{{\rm{b}},i} - \alpha_{{\rm s},i}G_{{\rm{s}},i}, i\in\mathcal{N}.
\end{align}
Note that $C_i$ can be positive (e.g., $G_{{\rm{b}},i}>0, G_{{\rm{s}},i}=0$), negative (e.g., $G_{{\rm{s}},i}>0,G_{{\rm{b}},i}=0$), or zero. In practice, to prevent any BS from buying the energy from the gird and then selling back to it to make non-justifiable profit which leads to energy inefficiency, the grid operator should set $\alpha_{{\rm{s}},i} \le \alpha_{{\rm{b}},i}, \forall i\in\mathcal{N}$; as a result, we can induce that at most one of $G_{{\rm b},i}$ and $G_{{\rm s},i}$ can be strictly positive, i.e., $G_{{\rm b},i} \cdot G_{{\rm s},i} = 0$ (otherwise, the cost in (\ref{total:cost}) can be further reduced by setting $G_{{\rm b},i}\leftarrow G_{{\rm b},i}- G_{{\rm s},i}$ and $G_{{\rm s},i}\leftarrow 0$ if $G_{{\rm b},i}\ge G_{{\rm s},i}>0$, or $G_{{\rm s},i}\leftarrow G_{{\rm s},i}- G_{{\rm b},i}$ and $G_{{\rm b},i}\leftarrow0$ if $G_{{\rm s},i}\ge G_{{\rm b},i}>0$). Moreover, the energy selling price is usually subject to a minimum value, given by $\alpha_{\min} > 0$, to encourage the renewable generation investment at the BSs; while the energy buying price cannot exceed the maximum electricity price in the grid, given by $\alpha_{\max} > 0$. Thus, we have
\begin{align}\label{eqn:price}
0 < \alpha_{\min} \le \alpha_{{\rm{s}},i} \le \alpha_{{\rm{b}},i}\le \alpha_{\max}, \forall i\in\mathcal{N}.
\end{align}

In cellular systems, the power consumption at each BS typically includes both the transmission power due to radio frequency (RF) power amplifiers (PAs), and the non-transmission power due to other components such as cooling systems, baseband units (BBU) for data processing, and circuits of RF chains (see Fig. \ref{fig:1}). We denote the radiated transmit power of each BS $i$ by $P_{{\rm t},i}\ge 0$, and generally model the non-transmission power as a constant denoted by $P_{{\rm c},i} > 0$. By combining them, we obtain the total power consumption at BS $i$, denoted by $P_i$, which should be no larger than the total energy available at BS $i$, i.e.,
\begin{align}
P_i = {P_{{\rm t},i}}/{\eta} + P_{{\rm c},i} \le E_i + G_{{\rm{b}},i} - G_{{\rm{s}},i},  i\in\mathcal{N},\label{eqn:power_consumption}
\end{align}
where $0 < \eta \le 1$ denotes the PA efficiency. Since $\eta$ is a constant, we normalize it as $\eta = 1$ unless stated otherwise.

\subsection{Downlink CoMP Transmission}

Next, we present the downlink CoMP transmission among the $N$ BSs in one cluster. We denote the channel vector from BS $i$ to MT $k$ as $\mv{h}_{ik} \in \mathbb{C}^{M\times 1}, i\in\mathcal{N},k\in\mathcal{K}$, and the channel vector from all $N$ BSs in the cluster to MT $k$ as $\mv{h}_k = [\mv{h}_{1k}^T~\ldots~\mv{h}_{Nk}^T]^T\in \mathbb{C}^{MN\times 1}, k\in\mathcal{K}$. We consider linear transmit beamforming applied at the BSs. Let the information signal for MT $k\in\mathcal{K}$ be denoted by $s_k$ and its associated beamforming vector across the $N$ BSs by $\mv{w}_k \in \mathbb{C}^{MN \times 1}$. Then the transmitted signal for MT $k$ can be expressed as
\begin{align*}
\mv{x}_k = \mv{w}_k s_k,
\end{align*}
where $s_k$ is assumed to be a complex random variable with zero mean and unite variance. Thus, the received signal at MT $k$ is given by
\begin{align*}
{y}_k = \mv{h}_k^H \mv{x}_k + \sum_{l\in\mathcal{K},l\neq k}\mv{h}_k^H\mv{x}_l+ v_k, k\in\mathcal{K},
\end{align*}
where $\mv{h}_k^H \mv{x}_k$ is the desired signal for MT $k$, $\sum_{l\in\mathcal{K},l\neq k}\mv{h}_k^H\mv{x}_l$ is the inter-user interference within the same cluster, and $v_k$ denotes the background additive white Gaussian noise (AWGN) at MT $k$, which may also include the downlink interference from other BSs outside this cluster. We assume that $v_k$'s are independent circularly symmetric complex Gaussian (CSCG) random variables each with zero mean and variance $\sigma_{k}^2$, i.e., $v_k\sim \mathcal{CN}(0, \sigma_{k}^2)$. Thus, the signal-to-interference-plus-noise-ratio (SINR) at MT $k$ can be expressed as
\begin{align*}
\mathtt{SINR}_k(\{\mv{w}_k\}) = \frac{|\mv{h}_k^H \mv{w}_k|^2 }{ \sum_{l\in\mathcal{K},l\neq k}|\mv{h}_k^H\mv{w}_l|^2 + \sigma_{k}^2}, k\in\mathcal{K}.
\end{align*}
The transmit power at each BS $i$, i.e., $P_{{\rm t}, i}$ in (\ref{eqn:power_consumption}), can be expressed as
\begin{align}\label{eqn:tx}
P_{{\rm{t}},i} =\sum_{k\in\mathcal{K}}\mv{w}_k^H\mv{B}_i\mv{w}_k, \forall i\in\mathcal{N},
\end{align}
where ${\mv{B}}_{i}\triangleq{\mathtt{Diag}}\bigg(\underbrace{0,\cdots,0}_{(i-1)M},\underbrace{1,\cdots,1}_{M},\underbrace{0,\cdots,0}_{(N-i)M}\bigg)$, with $\mathtt{Diag}(\mv{a})$ denoting a diagonal matrix with the diagonal elements given in the vector $\mv a$. In addition, we assume that the maximum transmit power at each BS $i$ is denoted by  $P_{{\rm max}, i} > 0$, and thus we have $P_{{\rm{t}},i} \le P_{{\rm max}, i}, \forall i\in\mathcal{N}$.

\subsection{Problem Formulation}
We aim to jointly optimize the $N$ BSs' purchased/sold energy from/to the grid, $\{G_{{\rm{b}},i}\}$ and $\{G_{{\rm{s}},i}\}$, and their cooperative transmit beamforming vectors, $\{\mv{w}_k\}$, so as to minimize the total energy cost of all $N$ BSs, i.e., $\sum_{i\in\mathcal{N}}C_i$ with $C_i$ given in (\ref{total:cost}), subject to each MT's QoS constraint that is specified by a minimum SINR requirement $\gamma_k$ for MT $k\in\mathcal{K}$. Here, the value of $\gamma_k$ should be set based on the service type (e.g., video call or online game) requested by each MT $k$. Mathematically, we formulate the joint energy trading and beamforming optimization problem as
\begin{align}
\mathrm{(P1)}:&\mathop{\mathtt{min}}\limits_{\{\mv{w}_k\},\{G_{{\rm{b}},i}\},\{G_{{\rm{s}},i}\}} \sum_{i\in\mathcal{N}}\left( \alpha_{{\rm b},i}G_{{\rm{b}},i} - \alpha_{{\rm s},i}G_{{\rm{s}},i}\right)\label{eqn:p1}\\
\mathtt{s.t.}~&\mathtt{SINR}_k(\{\mv{w}_k\})\ge \gamma_k, \forall k\in\mathcal{K}\label{eqn:7:final}\\
&\sum_{k\in\mathcal{K}}\mv{w}_k^H\mv{B}_i\mv{w}_k+ P_{{\rm c},i} \le E_i + G_{{\rm{b}},i} - G_{{\rm{s}},i}, \forall i\in\mathcal{N}
\label{eqn:8}\\
&\sum_{k\in\mathcal{K}}\mv{w}_k^H\mv{B}_i\mv{w}_k \le P_{{\rm max}, i}, \forall i\in\mathcal{N}\label{eqn:MaxPower}\\
&G_{{\rm{b}},i} \ge 0, G_{{\rm{s}},i} \ge 0, \forall i\in\mathcal{N},\label{eqn:11}
\end{align}
where (\ref{eqn:7:final}) denotes the set of QoS constraints for the $K$ MTs, (\ref{eqn:8}) specifies the power constraints at the $N$ BSs by combining (\ref{eqn:power_consumption}) and (\ref{eqn:tx}), and (\ref{eqn:MaxPower}) is for the individual maximum transmit power constraint at each of the $N$ BSs. In problem (P1), we do not explicitly add the constraint that at most one of $G_{{\rm{b}},i}$ and $G_{{\rm{s}},i}$ can be strictly positive for BS $i$, i.e., $G_{{\rm{b}},i}\cdot G_{{\rm{s}},i} = 0,\forall i\in\mathcal N$; however, it will be shown that the optimal solution to problem (P1) always satisfies such constraints, and thus there is no
loss of optimality by removing these constraints. Notice that problem (P1) is in general non-convex due to the non-convex QoS constraints in (\ref{eqn:7:final}).

Before solving (P1), we first check its feasibility as follows. Note that given any transmit beamforming vectors $\{\mv w_k\}$ satisfying (\ref{eqn:MaxPower}), each BS can always purchase sufficiently large amount of energy from the grid to satisfy the constraints in (\ref{eqn:8}) and (\ref{eqn:11}), e.g., by setting $G_{{\rm{b}},i} = (P_{{\rm max},i}+P_{{\rm c},i}-E_i)^+$ and $G_{{\rm{s}},i} = 0, \forall i\in \mathcal N$. As a result, to check the feasibility of (P1), we only need to check whether the constraints in (\ref{eqn:7:final}) and (\ref{eqn:MaxPower}) can be ensured at the same time. This is equivalent to solving the following feasibility problem to determine whether the $N$ BSs can use their individual power to meet the QoS constraints for all the $K$ MTs:
\begin{align}
\mathtt{find}~&  \{\mv{w}_k\} \nonumber\\
\mathtt{s.t.}~&\mathtt{SINR}_k(\{\mv{w}_k\})\ge \gamma_k, \forall k\in\mathcal{K}\nonumber\\
& \sum_{k\in\mathcal{K}}\mv{w}_k^H\mv{B}_i\mv{w}_k \le P_{{\rm max}, i}, \forall i\in\mathcal{N}.\label{eqn:p1:feasibility}
\end{align}
Problem (\ref{eqn:p1:feasibility}) has been solved by the standard convex optimization techniques via reformulating it as a second-order cone program (SOCP) \cite{WieselEldarShamai2006} or by the uplink-downlink duality based fixed point iteration algorithm \cite{Bengtsson1999,YuLan2007}. In the rest of this paper, we focus on the case that (P1) (or equivalently problem (\ref{eqn:p1:feasibility})) is feasible unless otherwise stated.

\section{Optimal Solution}\label{sec:optimal}

In this section, we present the optimal solution to problem (P1) by proposing an algorithm based on the Lagrange duality method \cite{BoydVandenberghe2004} and the uplink-downlink duality technique \cite{YuLan2007}. 

Let the dual variable associated with the $i$th power constraint in (\ref{eqn:8}) be denoted by $\mu_i\ge 0$, and that corresponding to the $i$th individual maximum transmit power constraint in (\ref{eqn:MaxPower}) by $\nu_i \ge 0$, $i\in\mathcal{N}$. Then we can express the partial Lagrangian of (P1) as
\begin{align}
&\mathcal{L}\left(\{\mv{w}_k\},\{G_{{\rm{b}},i}\},\{G_{{\rm{s}},i}\},\{\mu_i\},\{\nu_i\}\right) \nonumber \\
= & \sum_{k\in\mathcal{K}}\mv{w}_k^H \mv{B}_{\mu,\nu}\mv{w}_k + \sum_{i\in\mathcal{N}}
(\alpha_{{\rm b},i} - \mu_i)G_{{\rm{b}},i} + \sum_{i\in\mathcal{N}} (\mu_i - \alpha_{{\rm s},i})G_{{\rm{s}},i} \nonumber \\ &+ \sum_{i\in\mathcal{N}}(P_{{\rm c},i}-E_i)\mu_i -\sum_{i\in\mathcal{N}}P_{{\rm max},i}\nu_i,
\label{eqn:Lagangian}
\end{align}
where $\mv{B}_{\mu,\nu} \triangleq \sum_{i\in\mathcal{N}}(\mu_i+\nu_i)\mv{B}_i$.
Accordingly, the dual function is given by
\begin{align}
g(\{\mu_i\},\{\nu_i\}) =& \nonumber \\ \mathop{\mathtt{min}}\limits_{\{\mv{w}_k\},\{G_{{\rm{b}},i}\ge 0\},\{G_{{\rm{s}},i}\ge 0\}}&\mathcal{L}\left(\{\mv{w}_k\},\{G_{{\rm{b}},i}\},\{G_{{\rm{s}},i}\},\{\mu_i\},\{\nu_i\}\right)\nonumber\\
\mathtt{s.t.}~~~~~~~~~&\mathtt{SINR}_k(\{\mv{w}_k\})\ge \gamma_k, \forall k\in\mathcal{K},\label{eqn:dual_function}
\end{align}
and thus the dual problem is expressed as
\begin{align}
\mathrm{(D1):}\mathop{\mathtt{max}}\limits_{\{\mu_i\ge 0,\nu_i\ge 0\}}g(\{\mu_i\},\{\nu_i\}).\label{eqn:dual_problem}
\end{align}

Since (P1) itself is non-convex, in general only weak duality holds between (P1) and its dual problem (D1), that is, the optimal value achieved by (D1) is generally a lower bound on that of (P1) \cite{BoydVandenberghe2004}. However, due to the specific structure of (P1), we can show that strong duality indeed holds between (P1) and (D1), as stated in the following proposition.

\begin{proposition}\label{proposition:strong:duality}
The optimal value achieved by (D1) is equal to that by (P1).
\end{proposition}
\begin{proof}
This proposition relies on the fact that (P1) can be recast as a convex optimization problem. It is observed that any phase rotation of $\{\mv{w}_k\}$ does not change the SINR in (\ref{eqn:7:final}) or the power consumption in (\ref{eqn:8}) and (\ref{eqn:MaxPower}), i.e., $\mathtt{SINR}_k(\{\mv{w}_k\}) = \mathtt{SINR}_k(\{\mv{w}_k e^{j \phi_k}\})$ and $\mv{w}_k^H\mv{B}_i\mv{w}_k = (\mv{w}_k e^{j \phi_k})^H\mv{B}_i(\mv{w}_k e^{j \phi_k}), \forall k\in\mathcal K, i\in \mathcal N$, where $\phi_k, k\in \mathcal K$, are arbitrary phases. Without loss of optimality, we can choose $\{\mv{w}_k\}$ such that $\mv{h}_k^H\mv{w}_k$ is real and $\mv{h}_k^H\mv{w}_k \ge 0, \forall k\in\mathcal{K}$. In this case, by denoting $\mv{W} = \left[\mv{w}_1,\ldots,\mv{w}_K\right]$, we can reformulate (P1) as
 \begin{align}
\mathrm{(P1}\mathrm{-Ref)}:&
\mathop{\mathtt{min}}\limits_{\{\mv{w}_k\},\{G_{{\rm{b}},i}\},\{G_{{\rm{s}},i}\}}~ \sum_{i\in\mathcal{N}}\left( \alpha_{{\rm b},i}G_{{\rm{b}},i} - \alpha_{{\rm s},i}G_{{\rm{s}},i}\right) \nonumber\\
\mathtt{s.t.}~~~~~~~~&\left\|\mv{h}_k^H\mv{W} ~\sigma_k\right\|\le \sqrt{1+\frac{1}{\gamma_k}}\mv{h}_k^H\mv{w}_k, \forall k\in\mathcal{K}\label{eqn:7:final:convex}\\
&(\ref{eqn:8}),~(\ref{eqn:MaxPower}), ~ {\rm and} ~ (\ref{eqn:11})\nonumber,
\end{align}
where $\left\|\mv{h}_k^H\mv{W} ~\sigma_k\right\| = \sqrt{\sum_{l\in\mathcal{K}}|\mv{h}_k^H\mv{w}_l|^2 + \sigma_{k}^2}, k\in\mathcal{K}$. Since the constraints in (\ref{eqn:7:final:convex}) specify a set of second-order cones and thus are convex \cite[Chapters 2.2.3 and 4.4.2]{BoydVandenberghe2004}, it is evident that problem (P1-Ref) is convex, given that its objective function and other constraints are all convex. Based on the equivalence between problem (P1) and its convex reformulation (P1-Ref), this proposition can be proved, for which the details are given in Appendix \ref{appendix:1}.
\end{proof}

According to Proposition \ref{proposition:strong:duality}, we can solve (P1) by equivalently solving (D1). In the following, we first solve the problem in (\ref{eqn:dual_function}) for obtaining $g(\{\mu_i\},\{\nu_i\})$ under any given $\{\mu_i\}$ and $\{\nu_i\}$ satisfying $\mu_i \ge 0$ and $\nu_i \ge 0, \forall i\in\mathcal{N}$, and then minimize $g(\{\mu_i\},\{\nu_i\})$ over $\{\mu_i\}$ and $\{\nu_i\}$.

\subsection{Solve Problem (\ref{eqn:dual_function}) for Obtaining $g(\{\mu_i\},\{\nu_i\})$}
First, we have the following lemma.
\begin{lemma}\label{lemma:1}
In order for $g(\{\mu_i\},\{\nu_i\}) $ to be bounded from below, i.e., $g(\{\mu_i\},\{\nu_i\})>-\infty$, it must hold that
\begin{align}
\alpha_{{\rm s},i} \le \mu_i \le \alpha_{{\rm b},i}, \forall i\in\mathcal{N}.
\end{align}
\end{lemma}
\begin{proof}
See Appendix \ref{appendix:2}.
\end{proof}
\begin{figure*}[t]
\centering
 \epsfxsize=1\linewidth
    \includegraphics[width=12cm]{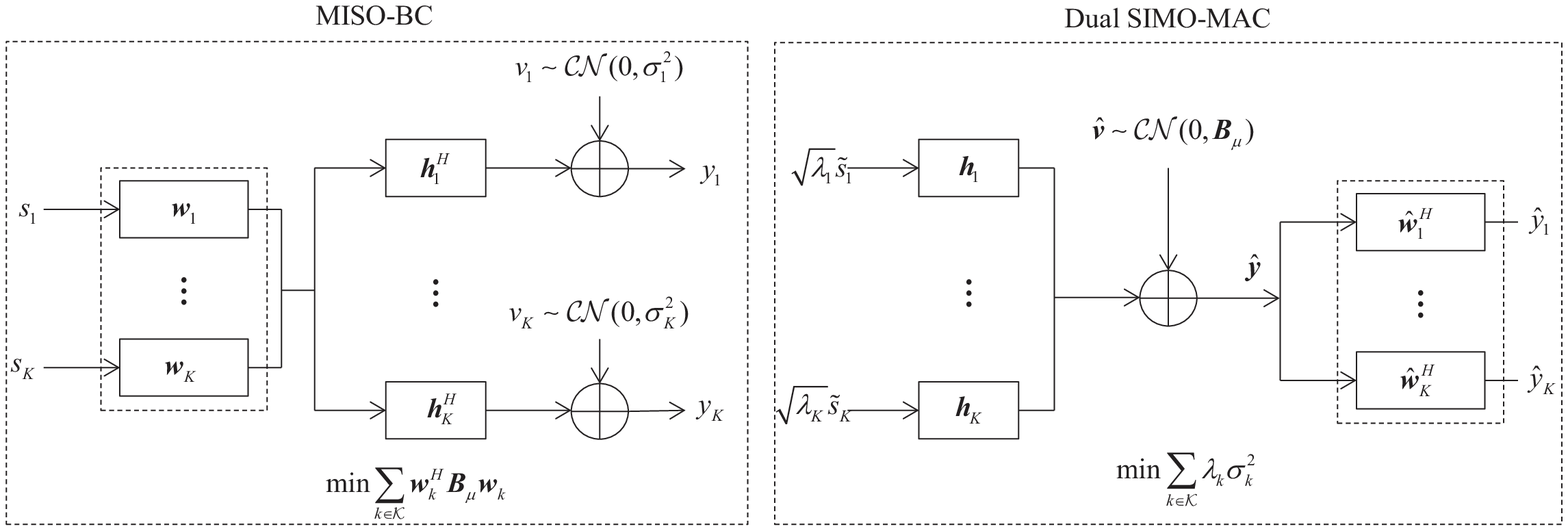}
\caption{Uplink-downlink duality for MISO-BC and SIMO-MAC.} \label{figure:bcmac}
\end{figure*}

According to Lemma \ref{lemma:1}, we only need to derive $g(\{\mu_i\},\{\nu_i\})$ for given $\{\mu_i\}$ and $\{\nu_i\}$ satisfying $0<\alpha_{{\rm s},i} \le \mu_i \le \alpha_{{\rm b},i}$ and $\nu_i \ge 0$, $\forall i\in\mathcal{N}$. In this case, we have $\mv{B}_{\mu,\nu} \succ \mv{0}$. Furthermore, we observe that the problem in (\ref{eqn:dual_function}) can be decomposed into the following $2N+1$ subproblems by dropping the irrelative term $\sum_{i\in\mathcal{N}}(P_{{\rm c},i}-E_i)\mu_i-\sum_{i\in\mathcal{N}}P_{{\rm max},i}\nu_i$:
\begin{align}
\mathop{\mathtt{min}}\limits_{\{\mv{w}_k\}}~&\sum_{k\in\mathcal{K}}\mv{w}_k^H \mv{B}_{\mu,\nu}\mv{w}_k\nonumber\\
\mathtt{s.t.}~~&\mathtt{SINR}_k(\{\mv{w}_k\})\ge \gamma_k, \forall k\in\mathcal{K}\label{eqn:dual_function1},\\
\mathop{\mathtt{min}}\limits_{G_{{\rm{b}},i}\ge 0}~&(\alpha_{{\rm b},i} - \mu_i)G_{{\rm{b}},i},~~\forall i\in\mathcal{N}\label{eqn:dual_function2},\\
\mathop{\mathtt{min}}\limits_{G_{{\rm{s}},i}\ge 0}~&(\mu_i - \alpha_{{\rm s},i})G_{{\rm{s}},i},~~\forall i\in\mathcal{N}\label{eqn:dual_function3},
\end{align}
where (\ref{eqn:dual_function2}) and (\ref{eqn:dual_function3}) each corresponds to $N$ subproblems (one for each BS $i$). For the subproblems in (\ref{eqn:dual_function2}) and (\ref{eqn:dual_function3}), it is easy to show that the optimal solutions are given by
\begin{align}
G_{{\rm{b}},i}^{\star} = G_{{\rm{s}},i}^{\star} = 0, \forall i\in\mathcal{N}.\label{eqn:30:solution}
\end{align}
Note that if $\mu_i = \alpha_{{\rm b},i}$ or $\mu_i = \alpha_{{\rm s},i}$ for any $i\in\mathcal{N}$, then the corresponding optimal solution of $G_{{\rm{b}},i}^{\star}$ or $G_{{\rm{s}},i}^{\star}$ in (\ref{eqn:30:solution}) is generally not unique and can take any nonnegative value. In this case, $G_{{\rm{b}},i}^{\star} = 0$ or $G_{{\rm{s}},i}^{\star} = 0$ is employed here for the purpose of  solving problem (\ref{eqn:dual_function2}) or (\ref{eqn:dual_function3}) to obtain the dual function only, while they may not be the optimal solution to the original problem (P1), as will be discussed later in Section \ref{sec:Minimize}.

Now, it only remains to solve problem (\ref{eqn:dual_function1}) with  $\mv{B}_{\mu,\nu} \succ \mv{0}$ for obtaining $g(\{\mu_i\},\{\nu_i\})$. To this end, we exploit the uplink-downlink duality as follows.

Problem (\ref{eqn:dual_function1}) can be viewed as a transmit beamforming problem for a multiple-input single-output broadcast channel (MISO-BC), as shown in the left sub-figure of Fig. \ref{figure:bcmac}, with the goal of minimizing the weighted sum-power $\sum_{k\in\mathcal{K}}\mv{w}_k^H\mv{B}_{\mu,\nu}\mv{w}_k$ at the transmitter subject to a set of SINR constraints $\{\gamma_k\}$. For the MISO-BC, its dual single-input multiple-output multiple-access channel (SIMO-MAC) is shown in the right sub-figure of Fig. \ref{figure:bcmac} by conjugating and transposing the channel vectors, where $K$ single-antenna transmitters send independent information to one common receiver with $MN$ antennas. For transmitter $k\in\mathcal{K}$, let $\lambda_k$ be its transmit power, $\tilde{s}_k$ denote its transmitted information signal with zero mean and unit variance, and $\mv{h}_k$ be its channel vector to the receiver in the dual SIMO-MAC. Then the received signal is expressed as
$\hat{\mv{y}} = \sum\limits_{k\in\mathcal{K}}\mv{h}_k\sqrt{\lambda_k}\tilde{s}_k + {\hat{\mv{v}}},$ where ${\hat{\mv{v}}}$ is a CSCG random vector with zero mean and covariance matrix $\mv{B}_{\mu,\nu}$ denoting the equivalent noise vector at the receiver, i.e., ${\hat{\mv{v}}}\sim \mathcal{CN}(\mv 0, \mv{B}_{\mu,\nu})$. By applying receive beamforming vector $\hat{\mv{w}}_k$'s, the SINRs of different users in the dual SIMO-MAC are then given by
\begin{align}\label{eqn:sinr:uplink}
&{\mathtt{SINR}}_k^{\mathrm{MAC}}(\{\hat{\mv{w}}_k,\lambda_k\}) \nonumber\\&= \frac{\lambda_k|\mv{h}_k^H\hat{\mv{w}}_k|^2}{ \sum_{l\in\mathcal{K},l\neq k} \lambda_l |\mv{h}_l^H\hat{\mv{w}}_k|^2 +  \hat{\mv{w}}_k^H\mv{B}_{\mu,\nu}\hat{\mv{w}}_k} ,\forall k \in \mathcal{K}.
\end{align}
The design objective for the dual SIMO-MAC is to minimize the weighted sum transmit power $\sum\limits_{k\in\mathcal{K}}\lambda_k\sigma_k^2$ by jointly optimizing the power allocation $\{\lambda_k\}$ and receive beamforming vectors $\{\hat{\mv{w}}_k\}$ subject to the same set of SINR constraints $\{\gamma_k\}$ as in the original MISO-BC given by (\ref{eqn:dual_function1}). We thus formulate the dual uplink problem as
\begin{align}
\mathop{\mathtt{min}}\limits_{\{\hat{\mv{w}}_k\},\{\lambda_k\ge 0\}} &  \sum_{k\in\mathcal{K}}\lambda_k\sigma_k^2 \nonumber \\
\mathtt{s.t.}~~~~& {\mathtt{SINR}}_k^{\mathrm{MAC}}(\{\hat{\mv{w}}_k,\lambda_k\})  \ge \gamma_k, \forall k\in\mathcal{K}.\label{eqn:dual_function3:uplink}
\end{align}
With $\mv{B}_{\mu,\nu} \succ \mv{0}$, it has been shown  in \cite{YuLan2007} that problems (\ref{eqn:dual_function1}) and (\ref{eqn:dual_function3:uplink}) are equivalent. Thus, we can solve the downlink problem (\ref{eqn:dual_function1}) by first solving the uplink problem (\ref{eqn:dual_function3:uplink}) and then mapping its solution to that of problem (\ref{eqn:dual_function1}), shown as follows.

First, consider the uplink problem (\ref{eqn:dual_function3:uplink}). Since it can be shown that the optimal solution of (\ref{eqn:dual_function3:uplink}) is always achieved when all the SINR constraints are met with equality \cite{YuLan2007}, it follows that the optimal uplink transmit power $\{\lambda^{\star}_k\}$ must be a fixed point solution of the following equations, and thus can be found via an iterative function evaluation procedure \cite{WieselEldarShamai2006}.
\begin{align}
\lambda_k^{\star}=\frac{1}{\left(1+\frac{1}{\gamma_k}\right) \mv{h}_k^H \left(\sum\limits_{l\in\mathcal{K}}\lambda_l^{\star}\mv{h}_l\mv{h}_l^H+\mv{B}_{\mu,\nu}  \right)^{-1}\mv{h}_k}, \forall k \in \mathcal{K}.\label{eqn:iterative:2}
\end{align}
With $\{\lambda^\star_k\}$ at hand, the optimal receive beamforming vector $\{\hat{\mv{w}}_k^\star\}$ can then be obtained based on the minimum-mean-squared-error (MMSE) principle as
\begin{align}
 \hat{{\mv{w}}}_k^{\star}~&=~\frac{\left(\sum\limits_{l\in\mathcal{K}}\lambda_l^{\star}\mv{h}_l\mv{h}_l^H +\mv{B}_{\mu,\nu} \right)^{-1}\mv{h}_k}{\left\|{\left(\sum\limits_{l\in\mathcal{K}}\lambda_l^{\star}\mv{h}_l\mv{h}_l^H +\mv{B}_{\mu,\nu}\right)^{-1}\mv{h}_k}\right\|}, \ \forall k \in \mathcal{K}. \label{eqn:iterative:1}
\end{align}

After obtaining the optimal solution of $\{\hat{\mv{w}}_k^\star\}$ and $\{\lambda_k^\star\}$ for the uplink problem (\ref{eqn:dual_function3:uplink}), we then map the solution to $\{{{\mv{w}}}_k^{\star}\}$ for the downlink problem (\ref{eqn:dual_function1}). As shown in \cite{YuLan2007}, $\{{{\mv{w}}}_k^{\star}\}$ and $\{\hat{\mv{w}}_k^\star\}$ are identical up to a certain scaling factor. Using this argument together with the fact that the optimal solution of (\ref{eqn:dual_function1}) is also attained with all the SINR constraints being tight similarly to that in problem (\ref{eqn:dual_function3:uplink}), it follows that $\{\mv{w}_k^\star\}$ can be obtained as $\mv{w}_k^\star = \sqrt{p_k^\star} \hat{\mv{w}}_k^\star, \forall k \in \mathcal{K}$, with $\mv{p}^\star = [p_1^\star,\ldots,p_{K}^\star]^T$ given by
\begin{align}
\mv{p}^\star = \bigg(\mv{I}-{\mv{D}}\left(\{\hat{\mv w}_k^\star,\gamma_k\}\right)\bigg)^{-1}\mv{u}\left(\{\hat{\mv w}_k^\star, \gamma_k\}\right),\label{solution}
\end{align}
where $[{\mv{D}}]_{kl}\left(\{\hat{\mv{w}}_k,\gamma_k\}\right) =\left\{\begin{array}{ll} 0, &
k=l \\ \frac{ \gamma_k|{\mv h}_k^H\hat{\mv w}_l|^2}{|{\mv h}_k^H\hat{\mv w}_k|^2}, & k\neq l \end{array} \right.$ and $\mv{u}\left(\{\hat{\mv{w}}_k,\gamma_k\}\right) = \left[\frac{\gamma_1\sigma_1^2}{|{\mv h}_1^H\hat{\mv w}_1|^2},\ldots,\frac{\gamma_{K}\sigma_{K}^2}{|{\mv h}_{K}^H\hat{\mv w}_{K}|^2}\right]^T$.

\subsection{Minimize $g(\{\mu_i\},\{\nu_i\})$ over $\{\mu_i\}$ and $\{\nu_i\}$}\label{sec:Minimize}

Up to now, we have obtained the optimal solution of $\{{{\mv{w}}}_i^{\star}\}, \{G_{{\rm{b}},i}^{\star}\}$, and $\{G_{{\rm{s}},i}^{\star}\}$ to the problem in (\ref{eqn:dual_function}) with given $\{\mu_i\}$ and $\{\nu_i\}$. Accordingly, the dual function $g(\{\mu_i\},\{\nu_i\}) $ has been obtained. Next, we solve problem (D1) by minimizing $g(\{\mu_i\},\{\nu_i\})$ over $\{\mu_i\}$ and $\{\nu_i\}$. Since $g(\{\mu_i\},\{\nu_i\})$ is convex but not necessarily differentiable, we can employ the ellipsoid method \cite{Boyd:ConvexII} to obtain the optimal $\{\mu_i^*\}$ and $\{\nu_i^*\}$ for (D1) by using the fact that the subgradients of $g(\{\mu_i\},\{\nu_i\})$ at given $\mu_i$ and $\nu_i$ can be shown to be $\sum_{k\in\mathcal{K}}\mv{w}_k^{\star H}\mv{B}_i\mv{w}_k^\star+ P_{{\rm c},i} - E_i - G_{{\rm{b}},i}^\star + G_{{\rm{s}},i}^\star = \sum_{k\in\mathcal{K}}\mv{w}_k^{\star H}\mv{B}_i\mv{w}_k^\star+ P_{{\rm c},i} - E_i$ and $\sum_{k\in\mathcal{K}}\mv{w}_k^{\star H}\mv{B}_i\mv{w}_k^\star - P_{\max,i}$, respectively, $i\in\mathcal{N}$.

With the obtained $\{\mu_i^*\}$ and $\{\nu_i^*\}$, the corresponding $\{{\mv{w}}_i^{\star}\}$  becomes the optimal transmit beamforming vectors for (P1), denoted by $\{\mv{w}^*_k\}$. However, the solutions of $\{G_{{\rm{b}},i}^{\star}\}$ and $\{G_{{\rm{s}},i}^{\star}\}$ given by (\ref{eqn:30:solution}) in general may not be the optimal solution to (P1), since they are not unique if $\alpha_{{\rm b},i} - \mu_i^* = 0$ or  $\mu_i^* - \alpha_{{\rm s},i} = 0$, for any $i\in\mathcal{N}$. Nevertheless, it can be easily checked that the optimal solution to (P1) is achieved when the constraints in (\ref{eqn:8}) are all met with equality. As a result, the optimal solution of $\{G_{{\rm{b}},i}^{*}\}$ and $\{G_{{\rm{s}},i}^{*}\}$ for (P1) can be obtained as
\begin{align}
G_{{\rm{b}},i}^*& = \left(\sum_{k\in\mathcal{K}}\mv{w}_k^{*H}\mv{B}_i\mv{w}_k^*+ P_{{\rm c},i} - E_i \right)^+, \forall i\in\mathcal{N}, \label{eqn:G_n1}\\
G_{{\rm{s}},i}^* &= \left(E_i-\sum_{k\in\mathcal{K}}\mv{w}_k^{*H}\mv{B}_i\mv{w}_k^*- P_{{\rm c},i} \right)^+, \forall i\in\mathcal{N}.\label{eqn:G_n2}
\end{align}
This result is intuitive, since if the amount of harvested energy by BS $i$, $E_i$, is smaller (or larger) than that of its consumed energy, $\sum_{k\in\mathcal{K}}\mv{w}_k^{*H}\mv{B}_i\mv{w}_k^*+ P_{{\rm c},i}$, then BS $i$ should purchase the insufficient energy (or sell the excess energy) from (or to) the grid.

To summarize, our proposed algorithm to solve (P1) is given in Table \ref{table2} as Algorithm 1.

\begin{table}[t]
\begin{center}
\caption{Algorithm for Solving Problem (P1)} \vspace{0.01cm}
 \hrule\label{table2}
\vspace{0.1cm} \textbf{Algorithm 1}  \vspace{0.1cm}
\hrule \vspace{0.01cm}
\begin{enumerate}
\item Initialize $\{\mu_i\}$ and $\{\nu_i\}$ with $\alpha_{{\rm s},i} \le \mu_i \le \alpha_{{\rm b},i}$ and $\nu_i\ge 0$, $\forall i\in\mathcal{N}$.
\item {\bf Repeat:}
    \begin{enumerate}
    \item Compute $\{\lambda_k^{\star}\}$ as a fixed point solution of (\ref{eqn:iterative:2}) by iterative function evaluation \cite{WieselEldarShamai2006}, and compute the uplink receive beamforming vectors $\{\hat{{\mv{w}}}_k^{\star}\}$ by (\ref{eqn:iterative:1}).
    \item Compute the downlink beamforming vectors as $\mv{w}_k^\star = \sqrt{p_k^\star} \hat{\mv{w}}_k^\star, \forall k \in \mathcal{K}$, with $\{p_k^{\star}\}$ given by (\ref{solution}).
    \item Compute the subgradients of $g(\{{\mu}_i\},\{{\nu}_i\})$ associated with $\mu_i$ and $\nu_i$ as $\sum_{k\in\mathcal{K}}\mv{w}_k^{\star H}\mv{B}_i\mv{w}_k^{\star} - E_i + P_{{\rm c},i}$ and $\sum_{k\in\mathcal{K}}\mv{w}_k^{\star H}\mv{B}_i\mv{w}_k^{\star} - P_{{\rm max},i}$, respectively, $i\in\mathcal{N}$, and then update $\{\mu_i\}$ and $\{\nu_i\}$ accordingly based on the ellipsoid method \cite{Boyd:ConvexII}, subject to $\alpha_{{\rm s},i} \le \mu_i \le \alpha_{{\rm b},i}$ and $\nu_i\ge 0$, $\forall i\in\mathcal{N}$.
    \end{enumerate}
\item {\bf Until} $\{\mu_i\}$ and $\{\nu_i\}$ all converge within a prescribed accuracy.
\item Set $\mv{w}_k^*=\mv{w}_k^\star,\forall k\in\mathcal{K}$.
\item Compute $\{G_{{\rm{b}},i}^{*}\}$ and $\{G_{{\rm{s}},i}^{*}\}$ given by (\ref{eqn:G_n1}) and (\ref{eqn:G_n2}).
\end{enumerate}
\vspace{0.01cm} \hrule \vspace{0.01cm}\label{algorithm:1}
\end{center}
\vspace{-0em}
\end{table}

\begin{remark}\label{remark:3.1}
For problem (P1), it follows that the optimal dual solution $\{\mu_i^*\}$ must satisfy that
\begin{eqnarray}\label{eqn:marginal_EnergyCost}
\mu_i^*\left\{\begin{array}{ll} = \alpha_{{\rm s},i}, & {\rm if}~G^*_{{\rm{b}},i}=0,
G^*_{{\rm{s}},i}>0, \\ \in [\alpha_{{\rm s},i},\alpha_{{\rm b},i}], &{\rm if}~
G^*_{{\rm{b}},i}=G^*_{{\rm{s}},i}=0,\\  = \alpha_{{\rm b},i}, &{\rm if}~
G^*_{{\rm{b}},i}>0,G^*_{{\rm{s}},i}=0, \end{array} \right. \forall i\in\mathcal N,
\end{eqnarray}
which can be explained intuitively by interpreting $\mu_i^*$  as the marginal energy cost for BS $i \in \mathcal N$. In other words, the marginal energy cost for BS $i$ (i.e., $\mu_i^*$) is equal to the unit energy selling/buying price $\alpha_{{\rm s},i}$/$\alpha_{{\rm b},i}$, if the BS sells/purchases energy to/from the grid (i.e., $G^*_{{\rm{s}},i}>0$ or $G^*_{{\rm{b}},i}>0$). Note that the set of marginal energy costs $\{\mu_i^*\}$ plays an important role in adjusting the cooperative transmit beamforming vectors at the $N$ BSs (cf. (\ref{eqn:iterative:1})). They allow the cooperative BSs to reallocate their power consumption pattern to follow the corresponding renewable generation profile, such that the BSs can better utilize the cheap renewable energy and thereby reduce the total energy cost, as will be validated in our simulation results in Section \ref{sec:numerical}.
\end{remark}

\section{Sub-optimal Solution with ZF Beamforming}\label{sec:suboptimal}

In the preceding section, we have obtained Algorithm 1 for optimally solving problem (P1). However, since Algorithm 1 requires both the inner fixed-point iteration and the outer ellipsoid iteration, it is in general of high implementation complexity. Therefore, we consider in this section a sub-optimal solution with lower complexity, which is based on cooperative ZF precoding/beamforming at the BSs \cite{WieselZF,ZhangRuiBD}. In this case, the transmit beamforming vectors should be designed to precancel any inter-user interference among different MTs, i.e., $\mv{h}_k^H \mv{w}_l = 0, \forall l,k\in\mathcal{K}, l\neq k$.{\footnote{It is assumed that all the channel vector $\mv{h}_k$'s are linearly independent in order for the ZF beamforming to be feasible.}} Note that the ZF beamforming is only applicable in the case when the number of MTs $K$ is no larger than the total number of transmit antennas at all $N$ BSs (i.e.,  $K\le MN$). Accordingly, the joint energy trading and cooperative ZF beamforming problem is formulated as
\begin{align}
\mathrm{(P1-ZF)}:&\mathop{\mathtt{min}}\limits_{\{\mv{w}_k\},\{G_{{\rm{b}},i}\},\{G_{{\rm{s}},i}\}}~  \sum_{i\in\mathcal{N}}\left( \alpha_{{\rm b},i}G_{{\rm{b}},i} - \alpha_{{\rm s},i}G_{{\rm{s}},i}\right)\nonumber\\
\mathtt{s.t.}~~&\frac{|\mv{h}_k^H \mv{w}_k|^2 }{\sigma_{k}^2}\ge \gamma_k, \forall k\in\mathcal{K}\label{eqn:SNR:zf}\\
&\mv{h}_k^H \mv{w}_l = 0, \forall l,k\in\mathcal{K}, l\neq k \label{eqn:zf}\\
&(\ref{eqn:8}),~(\ref{eqn:MaxPower}),~{\rm{and}}~(\ref{eqn:11}),\nonumber
\end{align}
where the SNR constraints in (\ref{eqn:SNR:zf}) are degenerated from the SINR constraints in (\ref{eqn:7:final}) due to the newly introduced ZF constraints in (\ref{eqn:zf}). It is worth noticing that with ZF beamforming, the BSs in general need higher transmit power to ensure the QoS constraints for all the MTs, as compared to the case with optimal beamforming in (P1). As a result, the feasibility of problem (P1-ZF) may not be always ensured under the condition that (P1) is feasible. Therefore, similar to problem (\ref{eqn:p1:feasibility}), we check the feasibility of (P1-ZF) by solving the following problem:
\begin{align}
\mathtt{find}~&  \{\mv{w}_k\} \nonumber\\
\mathtt{s.t.}~&\frac{|\mv{h}_k^H \mv{w}_k|^2 }{\sigma_{k}^2}\ge \gamma_k, \forall k\in\mathcal{K}\nonumber\\
&\mv{h}_k^H \mv{w}_l = 0, \forall l,k\in\mathcal{K}, l\neq k \nonumber\\
&\sum_{k\in\mathcal{K}}\mv{w}_k^H\mv{B}_i\mv{w}_k \le P_{{\rm max}, i}, \forall i\in\mathcal{N}.\label{eqn:p1ZF:feasibility}
\end{align}
Problem (\ref{eqn:p1ZF:feasibility}) has been solved in \cite{WieselZF}. In the rest of this section, we assume that problem (P1-ZF) (or equivalently problem (\ref{eqn:p1ZF:feasibility})) is feasible.

Next, we present the optimal solution to (P1-ZF) in closed-form. To start with, we first design the beamforming vectors $\{\mv{w}_k\}$ to remove the ZF constraints in (\ref{eqn:zf}) as follows. Define $\mv{H}_{-k} = \left[\mv{h}_{1}, \ldots, \mv{h}_{k-1}, \mv{h}_{k+1},\right.$ $\left.\ldots, \mv{h}_{K}\right]^H$ with $\mv{H}_{-k}\in\mathbb{C}^{(K-1)\times MN}$, $k\in\mathcal{K}$. Let the (reduced) singular value decomposition (SVD) of  $\mv{H}_{-k}$ be denoted as $\mv{H}_{-k} = \mv{U}_{k}\mv{\Sigma}_{k}\mv{V}_{k}^H$, where $ \mv{U}_{k}\in\mathbb{C}^{(K-1)\times (K-1)}$ with $ \mv{U}_{k} \mv{U}_{k}^H =  \mv{U}_{k}^H\mv{U}_{k} = \mv{I}$, $\mv{V}_{k}\in\mathbb{C}^{MN\times (K-1)}$ with $ \mv{V}_{k}^H\mv{V}_{k} = \mv{I}$, and $\mv{\Sigma}_{k}$ is a $(K-1)\times(K-1)$ diagonal matrix. Define the projection matrix $\mv{P}_{k} = \mv{I} - \mv{V}_{k}\mv{V}_{k}^H$. Without loss of generality, we can express $\mv{P}_{k} = \tilde{\mv{V}}_{k}\tilde{\mv{V}}_{k}^{H}$, where $\tilde{\mv{V}}_{k} \in \mathbb{C}^{MN \times (MN-K+1)}$ satisfies $\tilde{\mv{V}}_{k}^{H}\tilde{\mv{V}}_{k}=\mv{I}$ and ${\mv{V}}_{k}^{H}\tilde{\mv{V}}_{k}=\mv{0}$. Note that  $\left[{\mv{V}}_{k},\tilde{\mv{V}}_{k}\right]$ forms an $MN\times MN$ unitary matrix.  As a result, we have the following lemma.
\begin{lemma}\label{lemma:3}
The optimal ZF transmit beamforming vectors for problem (P1-ZF) are given as
\begin{align}
\mv{w}_{k} = \tilde{\mv{V}}_{k} \tilde{\mv{w}}_{k}, \forall k\in\mathcal{K},\label{eqn:optimalZF}
\end{align}
where $\tilde{\mv{w}}_k \in \mathbb{C}^{(MN-K+1)\times1}, \forall k\in\mathcal{K}$.
\end{lemma}
\begin{proof}
See Appendix \ref{appendix:3}.
\end{proof}

From Lemma \ref{lemma:3}, problem (P1-ZF) is accordingly reformulated as
\begin{align}
&\mathop{\mathtt{min}}\limits_{\{\tilde{\mv{w}}_k\},\{G_{{\rm{b}},i}\},\{G_{{\rm{s}},i}\}}~  \sum_{i\in\mathcal{N}}\left( \alpha_{{\rm b},i}G_{{\rm{b}},i} - \alpha_{{\rm s},i}G_{{\rm{s}},i}\right)\label{eqn:ZF:Reformulate:0}\\
&\mathtt{s.t.}~\frac{|\mv{h}_k^H\tilde{\mv{V}}_{k} \tilde{\mv{w}}_{k}|^2}{\sigma_{k}^2}\ge \gamma_k, \forall k\in\mathcal{K}\label{eqn:8:ZF:C1}\\
&~~~~~\sum_{k\in\mathcal{K}}\tilde{\mv{w}}_{k}^H\tilde{\mv{V}}_{k} ^H\mv{B}_i\tilde{\mv{V}}_{k} \tilde{\mv{w}}_{k}+ P_{{\rm c},i} \le E_i + G_{{\rm{b}},i} - G_{{\rm{s}},i}, \forall i\in\mathcal{N}
\label{eqn:8:ZF}\\
&~~~~~\sum_{k\in\mathcal{K}}\tilde{\mv{w}}_{k}^H\tilde{\mv{V}}_{k} ^H\mv{B}_i\tilde{\mv{V}}_{k} \tilde{\mv{w}}_{k}\le P_{{\rm max},i}, \forall i\in\mathcal{N}
\label{eqn:MaxPower:ZF}\\
&~~~~~(\ref{eqn:11}).\nonumber
\end{align}

Problem (\ref{eqn:ZF:Reformulate:0}) is non-convex due to the non-convexity of constraint (\ref{eqn:8:ZF:C1}). Nevertheless, due to the fact that any phase rotation of ${\mv{w}}_{k}$ does not change the SNR at MT $k\in\mathcal{K}$ in (\ref{eqn:8:ZF:C1}) and the power constraints in (\ref{eqn:8:ZF}), we can choose ${\mv{w}}_{k}$ such that $\mv{h}_k^H{\mv{w}}_{k}$ is real and $\mv{h}_k^H{\mv{w}}_{k} \ge 0, \forall k\in\mathcal{K}$. Accordingly, problem (\ref{eqn:ZF:Reformulate:0}) can be reformulated as the following convex optimization problem:
\begin{align}
\mathop{\mathtt{min}}\limits_{\{\tilde{\mv{w}}_k\},\{G_{{\rm{b}},i}\},\{G_{{\rm{s}},i}\}}~&  \sum_{i\in\mathcal{N}}\left( \alpha_{{\rm b},i}G_{{\rm{b}},i} - \alpha_{{\rm s},i}G_{{\rm{s}},i}\right)\label{eqn:ZF:Reformulate}\\
\mathtt{s.t.}~~~~~~~&\mv{h}_k^H\tilde{\mv{V}}_{k} \tilde{\mv{w}}_{k}\ge \sigma_{k}\sqrt{\gamma_k}, \forall k\in\mathcal{K}\label{eqn:zf:reformulation}\\
&(\ref{eqn:8:ZF}),(\ref{eqn:MaxPower:ZF}),~{\rm and}~(\ref{eqn:11}).\nonumber
\end{align}
By applying the Lagrange duality method \cite{BoydVandenberghe2004}, we can have the closed-form optimal solution to problem (\ref{eqn:ZF:Reformulate}), as stated in the following proposition.
\begin{proposition}\label{proposition:P5}
The optimal solution to problem (\ref{eqn:ZF:Reformulate}) is given by
\begin{align}
\tilde{\mv{w}}_k^{**} = & \frac{\sigma_{k}\sqrt{\gamma_k}}{\left|\mv{h}_k^H\tilde{\mv{V}}_{k}\left(\tilde{\mv{V}}_{k} ^H\mv{B}_{\mu^{**},\nu^{**}}\tilde{\mv{V}}_{k}\right)^{-1}\tilde{\mv{V}}_{k}^H\mv{h}_k\right|}\nonumber\\
&\cdot  \left(\tilde{\mv{V}}_{k} ^H\mv{B}_{\mu^{**},\nu^{**}}\tilde{\mv{V}}_{k}\right)^{-1}\tilde{\mv{V}}_{k}^H\mv{h}_k , \forall k\in\mathcal{K}\label{eqn:w_k_P5}\\
G_{{\rm{b}},i}^{**} = & \left(\sum_{k\in\mathcal{K}}\tilde{\mv{w}}_k^{**H}\tilde{\mv{V}}_{k}^H\mv{B}_i\tilde{\mv{V}}_{k}\tilde{\mv{w}}_k^{**}+ P_{{\rm c},i} - E_i\right)^+, \forall i\in\mathcal{N} \label{eqn:G_n1:P5} \\
G_{{\rm{s}},i}^{**} =& \left(E_i -\sum_{k\in\mathcal{K}}\tilde{\mv{w}}_k^{**H}\tilde{\mv{V}}_{k}^H\mv{B}_i\tilde{\mv{V}}_{k}\tilde{\mv{w}}_k^{**}- P_{{\rm c},i}\right)^+ , \forall i\in\mathcal{N},\label{eqn:G_n2:P5}
\end{align}
where $\mv{B}_{\mu^{**},\nu^{**}} \triangleq \sum_{i\in\mathcal{N}}(\mu_i^{**}+\nu_i^{**})\mv{B}_i \succ \mv{0}$ with $\{\mu_i^{**}\}$ and $\{\nu_i^{**}\}$ each denoting $N$ non-negative constants (dual variables) corresponding to the power constraints in (\ref{eqn:8:ZF}) and (\ref{eqn:MaxPower:ZF}), respectively. The dual variables $\{\mu_i^{**}\}$ and $\{\nu_i^{**}\}$ can be obtained by solving the dual problem of (\ref{eqn:ZF:Reformulate}) via the ellipsoid method,{\footnote{Please refer to Appendix \ref{appendix:4_1} for the detailed algorithm for obtaining $\{\mu_i^{**}\}$ and $\{\nu_i^{**}\}$.}} with $0<\alpha_{{\rm s}, i} \le \mu^{**}_i \le \alpha_{{\rm b}, i}$ and $\nu_i^{**}\ge 0$, $\forall i\in\mathcal{N}$.
\end{proposition}
\begin{proof}
See Appendix \ref{appendix:4_1}.
\end{proof}

By combining Proposition \ref{proposition:P5}  and Lemma \ref{lemma:3}, we can then obtain the optimal solution to problem (P1-ZF) as follows.
\begin{proposition}\label{corollary:P4}
The optimal solution to problem (P1-ZF) is given as
\begin{align}
{\mv{w}}_k^{**} = & \frac{\sigma_{k}\sqrt{\gamma_k}}{\left|\mv{h}_k^H\tilde{\mv{V}}_{k}\left(\tilde{\mv{V}}_{k} ^H\mv{B}_{\mu^{**},\nu^{**}}\tilde{\mv{V}}_{k}\right)^{-1}\tilde{\mv{V}}_{k}^H\mv{h}_k\right|} \nonumber\\& \cdot  \tilde{\mv{V}}_{k}\left(\tilde{\mv{V}}_{k} ^H\mv{B}_{\mu^{**},\nu^{**}}\tilde{\mv{V}}_{k}\right)^{-1}\tilde{\mv{V}}_{k}^H\mv{h}_k , \forall k\in\mathcal{K}\label{eqn:w_k_P4}\\
G_{{\rm{b}},i}^{**} =& \left(\sum_{k\in\mathcal{K}}{\mv{w}}_k^{**H}\mv{B}_i{\mv{w}}_k^{**}+ P_{{\rm c},i} - E_i\right)^+, \forall i\in\mathcal{N} \label{eqn:G_n1:P4} \\
G_{{\rm{s}},i}^{**} = &\left(E_i -\sum_{k\in\mathcal{K}}{\mv{w}}_k^{**H}\mv{B}_i{\mv{w}}_k^{**}- P_{{\rm c},i}\right)^+ , \forall i\in\mathcal{N}.\label{eqn:G_n2:P4}
\end{align}
\end{proposition}

Note that $\mu^{**}_i$ can be interpreted as the marginal energy cost for BS $i \in \mathcal N$, similar to that in Remark \ref{remark:3.1}. Therefore, based on $\{\mu^{**}_i\}$, the optimal ZF beamforming vectors $\{{\mv{w}}_k^{**}\}$ are adjusted such that the power consumption pattern among the $N$ BSs can efficiently follow the corresponding renewable generation profile for saving the total energy cost.

Finally, we make a complexity comparison between the sub-optimal solution by solving (P1-ZF) (i.e., Proposition \ref{corollary:P4}) versus the optimal solution by solving (P1) (i.e., Algorithm 1). Since the ellipsoid method is applied to obtain the optimal dual variables $\{\mu^{**}_i\}$ and $\{\nu^{**}_i\}$ for (P1-ZF) as well as $\{\mu^{*}_i\}$ and $\{\nu^{*}_i\}$ for (P1), a cutting-plane based iteration is required for both solutions, which has the complexity of order $\mathcal O(N^2)$ \cite{Boyd:ConvexII}. Nevertheless, within each iteration, the sub-optimal solution only requires to calculate one single expression in closed-form (cf. (\ref{eqn:w_k}) in Appendix \ref{appendix:4_1}), the complexity of which is of order $\mathcal O(KM^3N^3)$; while the optimal solution needs to employ an inner fixed-point iteration (i.e., (\ref{eqn:iterative:2})) to update the transmit power and beamforming vectors (cf. Algorithm 1), where the computational complexity of one inner iteration is of order $\mathcal O(KM^3N^3)$.{\footnote{It is difficult to analytically quantify the number of iterations required for each inner fixed-point iteration to converge. In general, the number of required iterations increases with $K$, $M$ and $N$.}} As a result, the sub-optimal solution significantly reduces the complexity as compared to the optimal solution.

\section{Simulation Results}\label{sec:numerical}

In this section, we provide simulation results to evaluate the performance of our proposed schemes with joint energy trading and communication cooperation based on the optimal solution in Section \ref{sec:optimal}, as well as the sub-optimal solution in Section \ref{sec:suboptimal} with ZF beamforming. For comparison, we also consider two conventional schemes that separately design communication cooperation and energy trading at BSs.

\subsection{Conventional Schemes}\label{Sec:Conventional}

\subsubsection{Optimal Beamforming to Minimize Sum-Power} In this scheme, the $N$ BSs first design the cooperative transmit beamforming so as to minimize their transmit sum-power subject to the QoS constraints given in (\ref{eqn:7:final}), for which the problem is formulated as
\begin{align}
\mathop{\mathtt{min}}\limits_{\{\mv{w}_k\}}~&  \sum_{k\in\mathcal{K}}\|\mv{w}_k\|^2\nonumber \\
\mathtt{s.t.}~~&\mathtt{SINR}_k(\{\mv{w}_k\})\ge \gamma_k, \forall k\in\mathcal{K}\nonumber\\
& \sum_{k\in\mathcal{K}}\mv{w}_k^H\mv{B}_i\mv{w}_k \le P_{{\rm max}, i}, \forall i\in\mathcal{N}.\label{eqn:OPTbeamforming}
\end{align}
Problem (\ref{eqn:OPTbeamforming}) can be solved by reformulating it as an SOCP\cite{WieselEldarShamai2006} or by the uplink-downlink duality based fixed point iteration algorithm \cite{YuLan2007}. Let the optimal beamforming solution to (\ref{eqn:OPTbeamforming}) be denoted by $\{\dot{\mv{w}}_i\}$. Next, each of the $N$ BSs sets the energy trading solution with the grid independently as $\dot G_{{\rm{b}},i} = \left(\sum\limits_{k\in\mathcal{K}}\dot{\mv{w}}_k^{H}\mv{B}_i\dot{\mv{w}}_k+ P_{{\rm c},i} - E_i\right)^+$ and $\dot G_{{\rm{s}},i}= \left(E_i-\sum\limits_{k\in\mathcal{K}}\dot{\mv{w}}_k^{{}H}\mv{B}_i\dot{\mv{w}}_k- P_{{\rm c},i}\right)^+ , \forall i\in\mathcal{N}$.
\subsubsection{ZF Beamforming to Minimize Sum-Power}
In this scheme, the $N$ BSs first cooperatively design their transmit ZF beamforming to minimize the transmit sum-power subject to the QoS constraints, for which the problem is formulated as
\begin{align}
\mathop{\mathtt{min}}\limits_{\{\mv{w}_k\}}~&  \sum_{k\in\mathcal{K}}\|\mv{w}_k\|^2 \nonumber \\
\mathtt{s.t.}~~&\frac{|\mv{h}_k^H \mv{w}_k|^2 }{\sigma_{k}^2}\ge \gamma_k, \forall k\in\mathcal{K}\nonumber\\
&\mv{h}_k^H \mv{w}_l = 0, \forall l,k\in\mathcal{K}, l\neq k\nonumber\\
& \sum_{k\in\mathcal{K}}\mv{w}_k^H\mv{B}_i\mv{w}_k \le P_{{\rm max}, i}, \forall i\in\mathcal{N}.\label{eqn:ZFbeamforming}
\end{align}
Problem (\ref{eqn:OPTbeamforming}) has been optimally solved in \cite{WieselZF,ZhangRuiBD}, for which the optimal ZF beamforming solution is denoted by $\{\ddot{\mv{w}}_k\}$. Next, each BS independently sets the energy trading solution as $\ddot G_{{\rm{b}},i} = \left(\sum_{k\in\mathcal{K}}\ddot{\mv{w}}_k^{ H}\mv{B}_i\ddot{\mv{w}}_k+ P_{{\rm c},i} - E_i\right)^+ $ and $\ddot G_{{\rm{s}},i}= \left(-\sum_{k\in\mathcal{K}}\ddot{\mv{w}}_k^{H}\mv{B}_i\ddot{\mv{w}}_k- P_{{\rm c},i} + E_i \right)^+ , \forall i\in\mathcal{N}$.

\begin{remark}
It is interesting to point out that for the special case of equal energy buying and selling prices for all BSs, i.e., $\alpha_{{\rm s},i}=\alpha_{{\rm b},i}=\alpha, \forall i\in\mathcal{N}$, the above conventional designs with optimal and ZF beamforming become the same beamforming and energy trading solutions to problems (P1) and (P1-ZF), respectively. However, for the general case of $\alpha_{{\rm s},i}<\alpha_{{\rm b},i}$ for any $i\in\mathcal N$, the conventional designs with optimal and ZF beamforming yield only suboptimal solutions for (P1) and (P1-ZF), respectively, as will be shown by simulation results next.
\end{remark}

\subsection{Performance Comparison}

\begin{figure*}
\centering
\subfigure[BS 1]{
\begin{minipage}[t]{0.31\linewidth}
\centering
\includegraphics[width=6.3cm]{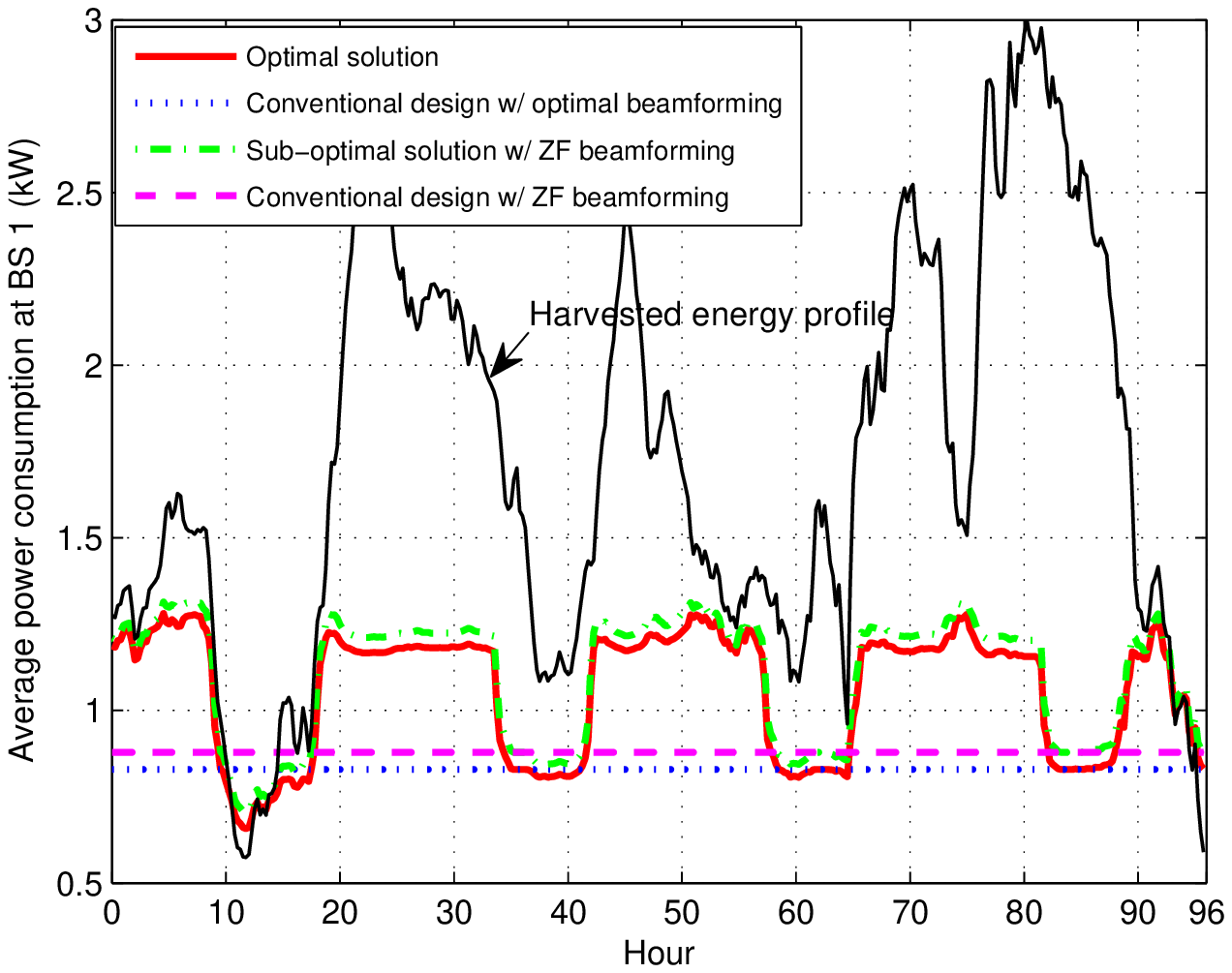}\vspace{-0em}
\end{minipage}
}\hspace{0ex}
\subfigure[BS 2]{
\begin{minipage}[t]{0.31\linewidth}
\centering
\includegraphics[width=6.3cm]{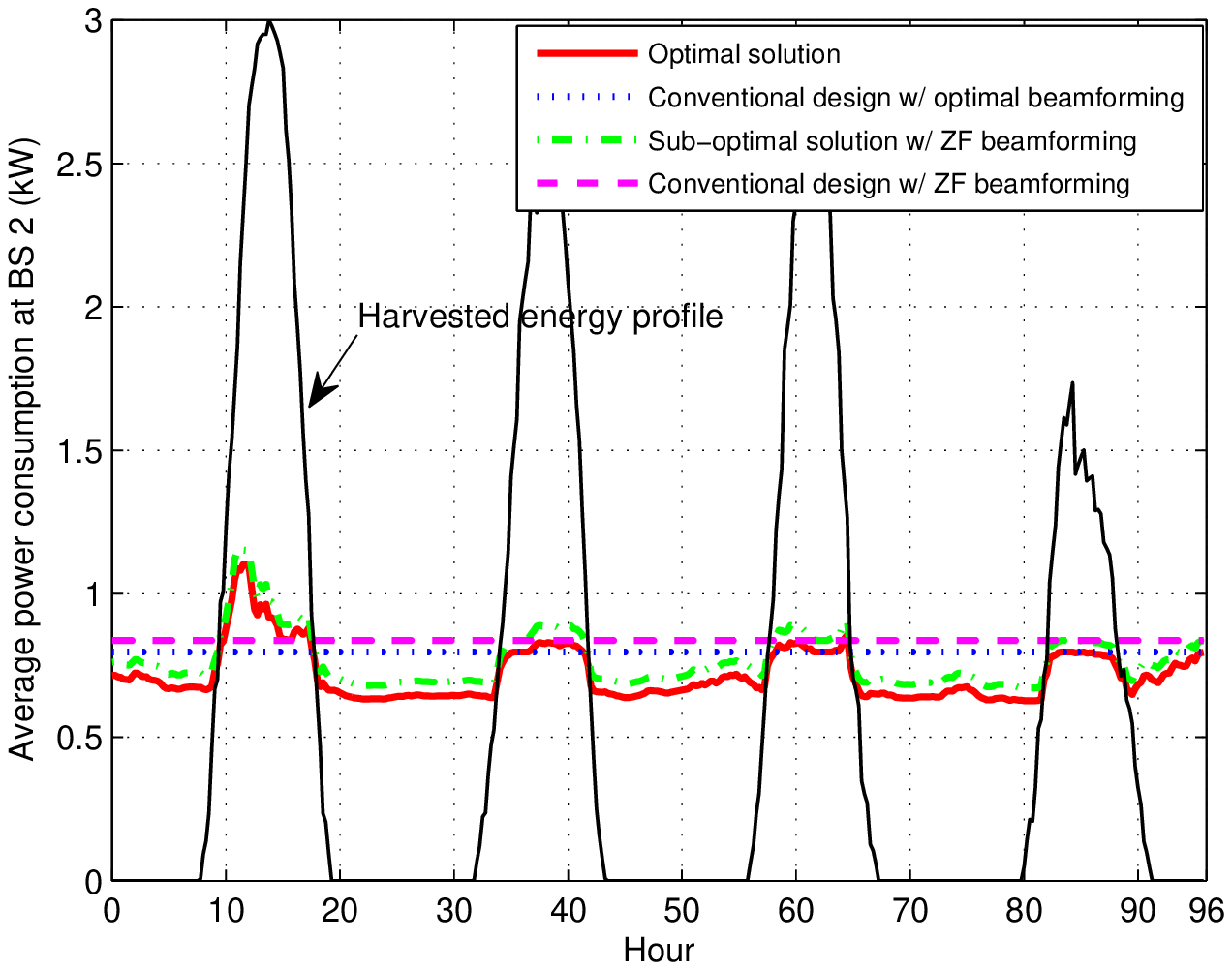}\vspace{-0em}
\end{minipage}
}\hspace{0ex}
\subfigure[BS 3]{
\begin{minipage}[t]{0.31\linewidth}
\centering
\includegraphics[width=6.3cm]{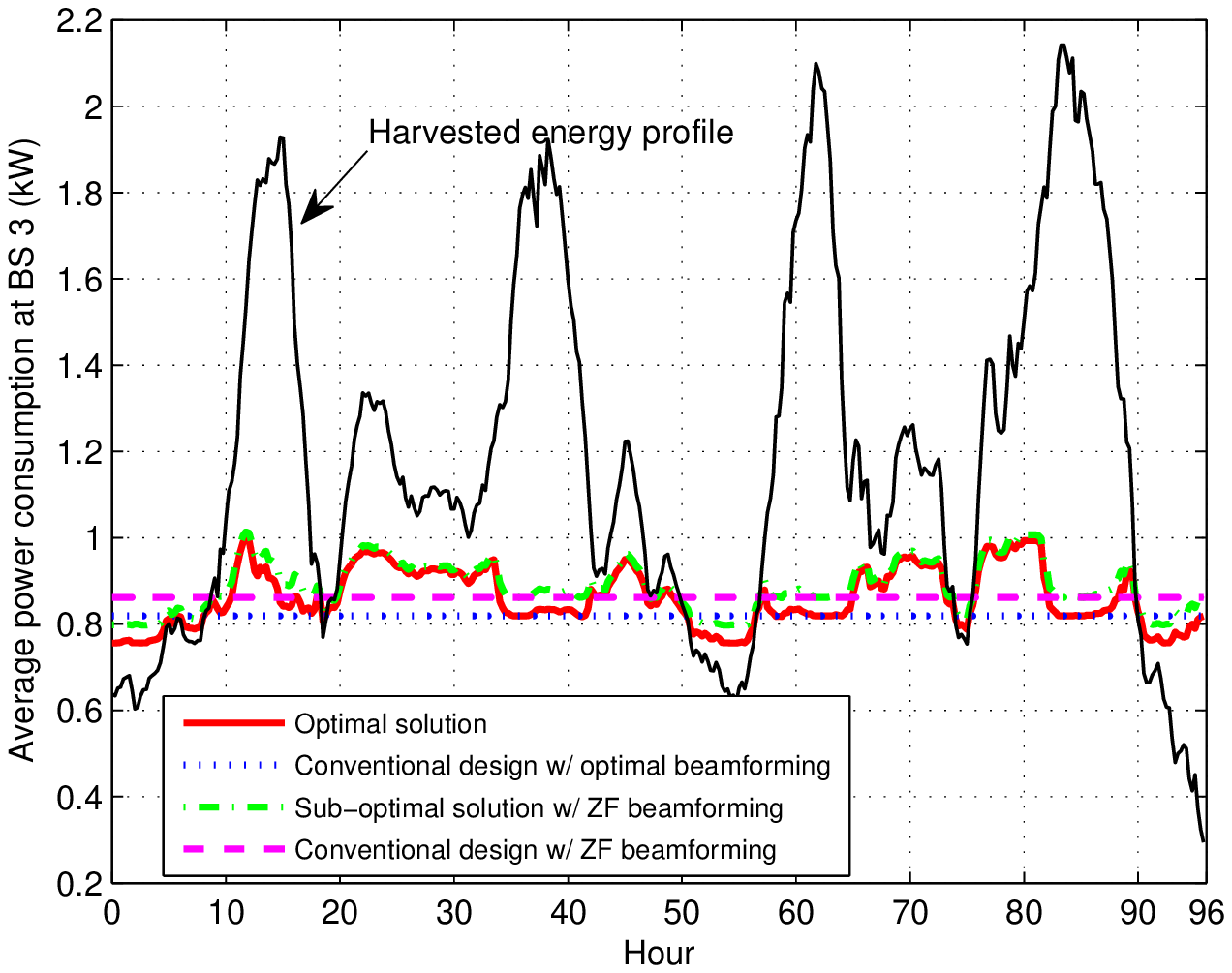}\vspace{-0em}
\end{minipage}
}\hspace{0ex}
\caption{The average power consumption at different BSs over time.} \label{fig:4}\vspace{-0em}
\end{figure*}

We consider a practical three-BS cluster (with $N=3$), where the cells are hexagonal with the inter-BS distance of one kilometer, the number of transmit antennas at each BS is $M=4$, and the total number of MTs is $K=8$. We assume that BS 1 and BS 2 are deployed with solar and wind generators, respectively, while BS 3 has both of them deployed. Then based on a real-world solar and wind energy production data,\footnote{See \url{http://www.elia.be/en/grid-data/power-generation/}.} we model the energy harvesting rates at the three BSs as shown in Fig. \ref{fig:4}, where the harvested energy at each BS has been averaged over 15 minutes, and thus there are 384 energy harvesting rate samples for each BS over 96 hours (i.e., four days). For each renewable energy sample, we apply the same set of 100 randomly generated user channels in order to focus our study on the impact of renewable generation variation. For each channel realization, we randomly generate the $8$ MTs in the three cells following a uniform distribution, and model each channel by a superposition of path loss and short-term Rayleigh fading. We assume that the background noise and the QoS requirement at each MT receiver are $\sigma_{k}^2 = -85$ dBm and $\gamma_k = 10$ dB, $\forall k\in\mathcal{K}$, respectively, while the PA efficiency, the maximum transmit power and the non-transmission constant power at each BS are $\eta = 0.1$ and $P_{{\rm max},i} = 100$ Watt (W), and $P_{{\rm c},i} = 500$ W, $\forall i \in \mathcal{N}$, respectively. We consider the prices for buying (selling) energy from (to) the grid as $\alpha_{{\rm b}, i} = 1$/kW ($\alpha_{{\rm s},i} = 0.1/$kW), $\forall i\in\mathcal{N}$, where the price unit is normalized without loss of generality. Note that among the randomly generated channels, we have only chosen the realizations such that problems (P1) and (P1-ZF) are both feasible for the ease of performance comparison.{\footnote{In particular, among the total 100 random channel realizations, problem (P1) is feasible for 95 realizations of them, while (P1-ZF) is feasible for 92 realizations. In other words, there are 3 channel realizations, for which (P1) is feasible but (P1-ZF) is not.}}

Fig. \ref{fig:4} shows the average power consumption at each of the three BSs over time, together with its harvested energy profile. For the two conventional designs with optimal and ZF beamforming, it is observed that the average power consumption at each BS remains constant over time, regardless of the fluctuations in energy harvesting rates at each BS. This is so because the two conventional designs solve the sum-power minimization problems (\ref{eqn:OPTbeamforming}) and (\ref{eqn:ZFbeamforming}) to obtain the respective transmit beamforming solutions by ignoring the two-way energy trading prices and the energy harvesting rates at all the BSs; as a result, the average power consumption at each BS is constant for this purposefully designed setup with fixed number of users and the same set of wireless channels over the time. In contrast, for the proposed optimal and sub-optimal solutions, the resulting average power consumption at each BS is observed to vary following a similar pattern as the corresponding energy harvesting rates. For example, during hours 0-9, BS 1 with large locally generated wind energy increases its transmission power and accordingly decreases the excess energy sold to the grid, while BS 2 and BS 3 with zero/smaller locally generated solar energy reduce their transmission power that need to be purchased from the grid, in order to minimize the total energy cost of the three BSs, given that $\alpha_{{\rm s}, i} < \alpha_{{\rm b},i}, \forall i\in\mathcal N$. Finally, it is observed that the optimal solution leads to small power consumption than the sub-optimal solution with ZF beamforming for each of the three BSs over the 96 hours. The energy (thus cost) reduction, however, is achieved at a cost of higher algorithm implementation complexity.

\begin{figure}
\centering
 \epsfxsize=1\linewidth
    \includegraphics[width=8cm]{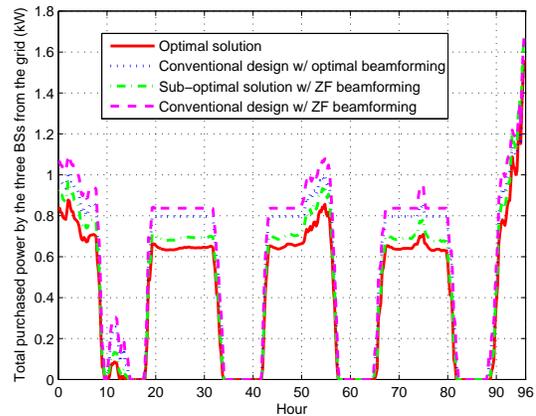}\vspace{-0em}
\caption{The total purchased power by the three BSs from the grid over time.} \label{fig:6}\vspace{-0em}
\end{figure}

\begin{figure}
\centering
 \epsfxsize=1\linewidth
    \includegraphics[width=8cm]{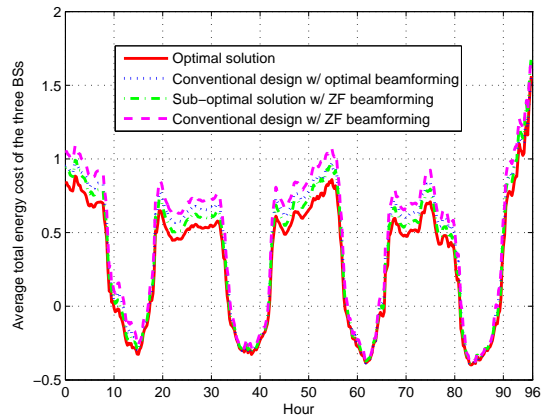}\vspace{-0em}
\caption{The total energy cost of the three BSs over time.} \label{fig:5}\vspace{-0em}
\end{figure}

Figs. \ref{fig:6} and \ref{fig:5} show the total purchased energy from the grid (i.e., $\sum_{i\in\mathcal N}G_{{\rm b},i}$) and the total energy cost of the three BSs over time, respectively, based on the results in Fig. \ref{fig:4}. It is observed that the proposed optimal and sub-optimal solutions reduce the expensive grid energy purchase (see Fig. \ref{fig:6}) and accordingly the total energy cost (see Fig. \ref{fig:5}) at all time, as compared to the conventional designs with optimal and ZF beamforming, respectively. The grid energy purchase and the energy cost reductions are also observed to be more substantial when the energy harvesting rates at different BSs are more unevenly distributed, e.g., during hours 0-9 and 20-30. On average, the average total energy costs over the entire period of 96 hours achieved by the four schemes of optimal solution, conventional design with optimal beamforming, sub-optimal solution with ZF beamforming, and conventional design with ZF beamforming are obtained as 0.3110, 0.3993, 0.3729, and 0.4601, respectively. As a result, the optimal solution and the sub-optimal solution with ZF beamforming achieve 22.12\% and 6.61\% total energy cost reductions over the conventional design with optimal beamforming, and 32.41\% and 18.96\% cost reductions over the conventional design with ZF beamforming, respectively.

\section{Conclusion}\label{sec:conclusion}

In this paper, we proposed a new cooperative energy trading approach for the downlink CoMP transmission powered by smart grids to reduce the energy cost of cellular systems. We minimize the total energy cost at the BSs in one CoMP cluster subject to the QoS constraints of MTs, by jointly optimizing the BSs' two-way energy trading with the grid in the supply side and their cooperative transmit beamforming in the demand side. We propose an optimal algorithm for this problem via applying techniques from convex optimization and uplink-downlink duality, as well as a sub-optimal solution with lower complexity based on ZF precoding. We show that interestingly, the conventional approach of minimizing the sum-power consumption of BSs is no more optimal when two-way energy trading with the grid is considered, while a joint energy trading and communication cooperation optimization is necessary to reschedule the power consumptions among BSs to follow the renewable energy patterns, so as to achieve the minimum energy cost for the system. It is our hope that the proposed cooperative energy trading and communication cooperation approach can open a new avenue for further research on the design of communication systems  powered by smart grids.

\appendix
\subsection{Proof of Proposition \ref{proposition:strong:duality}}\label{appendix:1}

Since (P1) can be recast as the equivalent form of (P1-Ref), it follows that $v_{\rm (P1)} = v_{\rm (P1-Ref)}$ with $v_{\rm (P1)}$ and $v_{\rm (P1-Ref)}$ denoting the optimal values of (P1) and (P1-Ref), respectively. Let $v_{\rm (D1)}$ denote the optimal value of (D1). We then prove Proposition \ref{proposition:strong:duality} by showing that $v_{\rm (D1)} = v_{\rm (P1-Ref)}$.

Let the dual variables associated with the power constraints in (\ref{eqn:8}) of (P1-Ref) be denoted by $\mu_i\ge 0, i\in\mathcal{N}$, and those corresponding to the individual maximum transmit power constraints in (\ref{eqn:MaxPower}) be denoted by $\nu_i\ge 0, i\in\mathcal{N}$. Then the partial Lagrangian of problem (P1-Ref) can also be expressed as $\mathcal{L}\left(\{\mv{w}_k\},\{G_{{\rm{b}},i}\},\{G_{{\rm{s}},i}\},\{\mu_i\},\{\nu_i\}\right)$ in (\ref{eqn:Lagangian}), and accordingly its dual function is given by
\begin{align*}
&g_{\rm Ref}(\{\mu_i\},\{\nu_i\}) = \nonumber\\
&\mathop{\mathtt{min}}\limits_{\{\mv{w}_k\},\{G_{{\rm{b}},i}\ge 0\},\{G_{{\rm{s}},i}\ge 0\}}\mathcal{L}\left(\{\mv{w}_k\},\{G_{{\rm{b}},i}\},\{G_{{\rm{s}},i}\},\{\mu_i\},\{\nu_i\}\right)\nonumber\\
&~~~~~~\mathtt{s.t.}~~~\left\|\mv{h}_k^H\mv{W} ~\sigma_k\right\|\le \sqrt{1+\frac{1}{\gamma_k}}\mv{h}_k^H\mv{w}_k, \forall k\in\mathcal{K}.
\end{align*}
The dual problem of (P1-Ref) is thus expressed as
\begin{align*}
\mathrm{(D1-Ref):}\mathop{\mathtt{max}}\limits_{\{\mu_i\ge 0,\nu_i\ge 0\}}g_{\rm Ref}(\{\mu_i\},\{\nu_i\}).
\end{align*}
Since (P1-Ref) is convex and satisfies the Slater's condition \cite{BoydVandenberghe2004}, we have $v_{\rm (P1-Ref)}=v_{\rm (D1-Ref)}$ with $v_{\rm (D1-Ref)}$ denoting the optimal value of (D1-Ref). Meanwhile, it is evident that $g_{\rm Ref}(\{\mu_i\},\{\nu_i\}) = g(\{\mu_i\},\{\nu_i\})$ holds for any given $\{\mu_i\}$ and $\{\nu_i\}$ due to the equivalent relationship between (\ref{eqn:7:final}) and (\ref{eqn:7:final:convex}). As a result, we can have $v_{\rm (D1)}=v_{\rm (D1-Ref)}=v_{\rm (P1-Ref)}$. Combining this argument with $v_{\rm (P1)} = v_{\rm (P1-Ref)}$, it follows that $v_{\rm (D1)}=v_{\rm (P1)}$. Therefore, Proposition \ref{proposition:strong:duality} is proved.

\subsection{Proof of Lemma \ref{lemma:1}}\label{appendix:2}
First, we assume that $\mu_i  < \alpha_{{\rm s},i}$ for any $i\in\mathcal{N}$. In this case, it follows that $\mathcal{L}\left(\{\mv{w}_k\},\{G_{{\rm{b}},i}\},\{G_{{\rm{s}},i}\},\{\mu_i\},\{\nu_i\}\right) \to -\infty$ if $G_{{\rm{s}},i} \to \infty$, and therefore $g(\{\mu_i\},\{\nu_i\}) $ is unbounded from below. As a result, $\mu_i  < \alpha_{{\rm s},i}$ cannot be true in order for $g(\{\mu_i\},\{\nu_i\}) $ to be bounded from below. It thus follows that $\mu_i \ge \alpha_{{\rm s},i}, \forall i\in\mathcal{N}$.

Next, we can similarly show that if $\mu_i > \alpha_{{\rm b},i}$ for any $i\in\mathcal{N}$, then $\mathcal{L}\left(\{\mv{w}_k\},\{G_{{\rm{b}},i}\},\{G_{{\rm{s}},i}\},\{\mu_i\},\{\nu_i\}\right) \to -\infty$ holds by setting $G_{{\rm{b}},i} \to \infty$, i.e., $g(\{\mu_i\},\{\nu_i\}) $ is unbounded from below. As a result, we have $\mu_i \le \alpha_{{\rm b},i}, \forall i\in\mathcal{N}$ in order for $g(\{\mu_i\},\{\nu_i\}) $ to be bounded from below.

By combining the two cases, Lemma \ref{lemma:1} follows immediately.

\subsection{Proof of Lemma \ref{lemma:3}}\label{appendix:3}

Without loss of generality, we can express $\mv{w}_{k}$ as
\begin{align}\label{eqn:zf:general}
\mv{w}_{k} = \left[{\mv{V}}_{k},\tilde{\mv{V}}_{k}\right] \left[\bar{\mv{w}}_{k}^T,\tilde{\mv{w}}_{k}^T\right]^T, \forall k\in\mathcal{K},
\end{align}
where $\bar{\mv{w}}_k \in \mathbb{C}^{(K-1)\times1}$.
By using (\ref{eqn:zf:general}) together with the fact that the ZF constraints in (\ref{eqn:zf}) are equivalent to $\mv{H}_{-k} \mv{w}_k = \mv{0}, \forall k\in\mathcal{K}$, we have
\begin{align*}
&\mv{H}_{-k}\mv{w}_{k} = \mv{H}_{-k}\left[{\mv{V}}_{k},\tilde{\mv{V}}_{k}\right] \left[\bar{\mv{w}}_{k}^T,\tilde{\mv{w}}_{k}^T\right]^T \\= & \mv{U}_{k}\mv{\Sigma}_{k}\mv{V}_{k}^H\left[{\mv{V}}_{k},\tilde{\mv{V}}_{k}\right] \left[\bar{\mv{w}}_{k}^T,\tilde{\mv{w}}_{k}^T\right]^T \\ = &\mv{U}_{k}\mv{\Sigma}_{k}\left[{\mv{I}},\tilde{\mv{0}}\right] \left[\bar{\mv{w}}_{k}^T,\tilde{\mv{w}}_{k}^T\right]^T = \mv{U}_{k}\mv{\Sigma}_{k}\bar{\mv{w}}_{k} = \mv{0}, \forall k\in\mathcal{K}.
\end{align*}
Since $\mv{U}_{k}\in \mathbb{C}^{(K-1)\times(K-1)}$ and $\mv{\Sigma}_{k}\in \mathbb{C}^{(K-1)\times(K-1)}$ are both of full rank provided that the channel vector $\mv{h}_k$'s are linearly independent, it follows that $\mv{U}_{k}\mv{\Sigma}_{k}$ is also of full rank. As a result, it must hold that $\bar{\mv{w}}_{k} = \mv{0}$. Together with (\ref{eqn:zf:general}), we then have $\mv{w}_{k} = \tilde{\mv{V}}_{k} \tilde{\mv{w}}_{k}, \forall k\in\mathcal{K}$, which completes the proof of Lemma \ref{lemma:3}.

\subsection{Proof of Proposition \ref{proposition:P5}}\label{appendix:4_1}

For problem (\ref{eqn:ZF:Reformulate}), let $\mu_i \ge 0$ and $\nu_i\ge 0$ denote the dual variables associated with $i$th power constraint in (\ref{eqn:8:ZF}) and $i$th individual maximum transmit power constraint in (\ref{eqn:MaxPower:ZF}), respectively, $i\in\mathcal{N}$. Then we have the Lagrangian of  (\ref{eqn:ZF:Reformulate}) as
\begin{align}
&\mathcal{L}_{\rm{ZF}}\left(\{\tilde{\mv{w}}_k\},\{G_{{\rm{b}},i}\},\{G_{{\rm{s}},i}\},\{\mu_i\},\{\nu_i\}\right)\nonumber\\
= & \sum_{k\in\mathcal{K}}\tilde{\mv{w}}_{k}^H\tilde{\mv{V}}_{k} ^H\mv{B}_{\mu,\nu}\tilde{\mv{V}}_{k}\tilde{\mv{w}}_k + \sum_{i\in\mathcal{N}}
(\alpha_{{\rm b},i} - \mu_i)G_{{\rm{b}},i} \nonumber\\&+ \sum_{i\in\mathcal{N}} (\mu_i - \alpha_{{\rm s},i})G_{{\rm{s}},i}+ \sum_{i\in\mathcal{N}}(P_{{\rm c},i}-E_i)\mu_i -\sum_{i\in\mathcal{N}}P_{{\rm max},i}\nu_i,
\label{eqn:Lagangian:zf}
\end{align}
where $\mv{B}_{\mu,\nu} = \sum_{i\in\mathcal{N}}(\mu_i+\nu_i)\mv{B}_i$.
Accordingly, the dual function can then be given by
\begin{align}
g_{\rm{ZF}}(\{\mu_i\},\{\nu_i\}) = &\nonumber\\ \mathop{\mathtt{min}}\limits_{\{\mv{w}_k\},\{G_{{\rm{b}},i}\ge 0\},\{G_{{\rm{s}},i}\ge 0\}}&\mathcal{L}_{\rm{ZF}}\left(\{\tilde{\mv{w}}_k\},\{G_{{\rm{b}},i}\},\{G_{{\rm{s}},i}\},\{\mu_i\},\{\nu_i\}\right).\nonumber\\
\mathtt{s.t.}~~~~~~&\mv{h}_k^H \tilde{\mv{V}}_{k} \tilde{\mv{w}}_{k}\ge \sigma_{k}\sqrt{\gamma_k}, \forall k\in\mathcal{K}.\label{eqn:dual_function:zf}
\end{align}
The dual problem of  (\ref{eqn:ZF:Reformulate}) is thus expressed as
\begin{align}
\mathop{\mathtt{max}}\limits_{\{\mu_i\ge 0,\nu_i\ge 0\}}g_{\rm{ZF}}(\{\mu_i\},\{\nu_i\}).\label{eqn:dual_problem:zf}
\end{align}
Since problem (\ref{eqn:ZF:Reformulate}) is convex and satisfies the Slater's conditions \cite{BoydVandenberghe2004}, strong duality holds between (\ref{eqn:ZF:Reformulate}) and (\ref{eqn:dual_problem:zf}). Accordingly, we can solve problem (\ref{eqn:ZF:Reformulate}) by equivalently solving  (\ref{eqn:dual_problem:zf}). In the following, we first solve the problem in (\ref{eqn:dual_function:zf}) to obtain $g_{\rm{ZF}}(\{\mu_i\},\{\nu_i\})$ under given $\{\mu_i\}$ and $\{\nu_i\}$, and then maximize $g_{\rm{ZF}}(\{\mu_i\},\{\nu_i\})$ over $\{\mu_i\}$ and $\{\nu_i\}$.

First, we have the following lemma for the dual function $g_{\rm{ZF}}(\{\mu_i\},\{\nu_i\})$.
\begin{lemma}\label{lemma:2}
In order for $g_{\rm{ZF}}(\{\mu_i\},\{\nu_i\})$ to be bounded from below, it must hold that
\begin{align}
\alpha_{{\rm s},i} \le \mu_i \le \alpha_{{\rm b},i}, \forall i\in\mathcal{N}.
\end{align}
\end{lemma}
\begin{proof}
The proof is similar to that of Lemma \ref{lemma:1}, and thus is omitted here for brevity.
\end{proof}
According to Lemma \ref{lemma:2}, we only need to derive $g_{\rm ZF}(\{\mu_i\},\{\nu_i\})$ for given $\{\mu_i\}$ and $\{\nu_i\}$ satisfying $\alpha_{{\rm s},i} \le \mu_i \le \alpha_{{\rm b},i}$ and $\nu_i\ge 0$, $\forall i\in\mathcal{N}$. As a consequence, it follows that $\mv{B}_{\mu,\nu} \succ \mv{0}$. In this case, we can obtain the optimal solution to the problem in (\ref{eqn:dual_function:zf}) in the following lemma.
\begin{lemma}\label{lemma:4.3}
The optimal solution to the problem in (\ref{eqn:dual_function:zf}) is given by $G_{{\rm{b}},i}^{{\star\star}} = G_{{\rm{s}},i}^{{\star\star}} = 0, \forall i\in\mathcal{N}$ and
\begin{align}
\tilde{\mv{w}}_k^{\star\star} = & \frac{\sigma_{k}\sqrt{\gamma_k}}{\left|\mv{h}_k^H\tilde{\mv{V}}_{k} \left(\tilde{\mv{V}}_{k} ^H\mv{B}_{\mu,\nu}\tilde{\mv{V}}_{k}\right)^{-1}\tilde{\mv{V}}_{k}^H\mv{h}_k\right|} \nonumber\\& \cdot   \left(\tilde{\mv{V}}_{k} ^H\mv{B}_{\mu,\nu}\tilde{\mv{V}}_{k}\right)^{-1}\tilde{\mv{V}}_{k}^H\mv{h}_k, \forall k\in\mathcal{K}.\label{eqn:w_k}
\end{align}
\end{lemma}
\begin{proof}
See Appendix \ref{appendix:4}.
\end{proof}

From Lemma \ref{lemma:4.3}, we have accordingly obtained $g_{\rm{ZF}}(\{\mu_i\},\{\nu_i\})$. Next, we minimize $g_{\rm{ZF}}(\{\mu_i\},\{\nu_i\})$ over $\{\mu_i\}$ and $\{\nu_i\}$. Similar to Section \ref{sec:Minimize}, we employ the ellipsoid method \cite{Boyd:ConvexII} to derive the optimal dual variables $\{\mu_i^{**}\}$ and $\{\nu_i^{**}\}$ by using the fact that the subgradients of $g_{\rm{ZF}}(\{\mu_i\},\{\nu_i\})$ for given $\mu_i$ and $\nu_i$ are given by $\sum_{k\in\mathcal{K}}\tilde{\mv{w}}_{k}^{{\star\star} H}\tilde{\mv{V}}_{k} ^H\mv{B}_i\tilde{\mv{V}}_{k} \tilde{\mv{w}}_{k}^{{\star\star}}+ P_{{\rm c},i} - E_i$ and $\sum_{k\in\mathcal{K}}\tilde{\mv{w}}_{k}^{{\star\star} H}\tilde{\mv{V}}_{k} ^H\mv{B}_i\tilde{\mv{V}}_{k} \tilde{\mv{w}}_{k}^{{\star\star}}-P_{{\rm max},i}$, respectively, $\forall i\in\mathcal{N}$.

With $\{\mu_i^{**}\}$ and $\{\nu_i^{**}\}$ obtained, the optimal transmit beamforming vector $\tilde{\mv{w}}_k^{**}$'s in (\ref{eqn:w_k_P5}) follow from Lemma \ref{lemma:4.3} by replacing $\{\mu_i\}$ and $\{\nu_i\}$ as $\{\mu_i^{**}\}$ and $\{\nu_i^{**}\}$. Accordingly, the optimal energy management strategy $\{G_{{\rm{b}},i}^{**}\}$ and $\{G_{{\rm{s}},i}^{**}\}$ in (\ref{eqn:G_n1:P5}) and in (\ref{eqn:G_n2:P5}) can be obtained based on the fact that the power constraints should be met with equality at the optimal solution of (\ref{eqn:ZF:Reformulate}). Therefore, Proposition \ref{proposition:P5} is proved.

For convenience, we summarize the algorithm to solve problem (\ref{eqn:ZF:Reformulate}) in Table \ref{table3} as Algorithm 2.

\begin{table}[t]
\begin{center}
\caption{{\bf Algorithm 2}: Algorithm for Solving Problem (\ref{eqn:ZF:Reformulate})} \vspace{0.01cm}
\hrule \vspace{0.01cm}\label{table3}
\begin{enumerate}
\item Initialize $\{\mu_i\}$ and $\{\nu_i\}$ with $\alpha_{{\rm s},i} \le \mu_i \le \alpha_{{\rm b},i}$ and $\nu_i \ge 0, \forall i\in\mathcal{N}$.
\item {\bf Repeat:}
    \begin{enumerate}
    \item Obtain $\{\tilde{\mv{w}}_{k}^{\star\star}\}$ given by (\ref{eqn:w_k}).
    \item Compute the subgradients of $g(\{{\mu}_i\},\{{\nu}_i\})$ associated with $\mu_i$ and $\nu_i$ as $\sum_{k\in\mathcal{K}}\tilde{\mv{w}}_{k}^{\star\star H}\tilde{\mv{V}}_{k} ^H\mv{B}_i\tilde{\mv{V}}_{k} \tilde{\mv{w}}_{k}^{\star\star} +P_{{\rm c},i}  - E_i$ and $\sum_{k\in\mathcal{K}}\tilde{\mv{w}}_{k}^{\star\star H}\tilde{\mv{V}}_{k} ^H\mv{B}_i\tilde{\mv{V}}_{k} \tilde{\mv{w}}_{k}^{\star\star} - P_{{\rm max},i}$, respectively, $i\in\mathcal{N}$, then update $\{\mu_i\}$ and $\{\nu_i\}$ accordingly based on the ellipsoid method \cite{Boyd:ConvexII}, subject to $\alpha_{{\rm s},i} \le \mu_i \le \alpha_{{\rm b},i},$ and $\nu_i \ge 0, \forall i\in\mathcal{N}.$
    \end{enumerate}
\item {\bf Until} $\{\mu_i\}$ and $\{\nu_i\}$ all converge within a prescribed accuracy.
\item Compute $\{\tilde{\mv{w}}_k^{**}\}$,  $\{G_{{\rm{b}},i}^{**}\}$ and $\{G_{{\rm{s}},i}^{**}\}$ given by (\ref{eqn:w_k_P5}),  (\ref{eqn:G_n1:P5}) and (\ref{eqn:G_n2:P5}), respectively.
\end{enumerate}
\vspace{0.01cm} \hrule \vspace{0.01cm}\label{algorithm:3}
\end{center}
\vspace{-0cm}
\end{table}

\subsection{Proof of Lemma \ref{lemma:4.3}}\label{appendix:4}

It is evident that the problem in (\ref{eqn:dual_function:zf}) can be decomposed into the following $K$ subproblems (each for one MT $k\in\mathcal{K}$) as well as subproblems (\ref{eqn:dual_function2}) and (\ref{eqn:dual_function3}) by dropping the irrelative term $\sum_{i\in\mathcal{N}}(P_{{\rm c},i}-E_i)\mu_i-\sum_{i\in\mathcal{N}}P_{{\rm max},i}\nu_i$.
\begin{align}
\mathop{\mathtt{min}}\limits_{\tilde{\mv{w}}_k}~&\tilde{\mv{w}}_{k}^H\tilde{\mv{V}}_{k} ^H\mv{B}_{\mu,\nu}\tilde{\mv{V}}_{k}\tilde{\mv{w}}_k\nonumber\\
\mathtt{s.t.}~&\mv{h}_k^H \tilde{\mv{V}}_{k} \tilde{\mv{w}}_{k}\ge \sigma_{k}\sqrt{\gamma_k}, \forall k\in\mathcal{K}.\label{eqn:dual_function1:ZF}
\end{align}
Since it is easy to show that the optimal solutions to the subproblems (\ref{eqn:dual_function2}) and (\ref{eqn:dual_function3}) are given in $G_{{\rm{b}},i}^{{\star\star}} = G_{{\rm{s}},i}^{{\star\star}} = 0, \forall i\in\mathcal{N}$ (see  (\ref{eqn:30:solution})), we only need to verify that the optimal solution of (\ref{eqn:dual_function1:ZF}) is given in (\ref{eqn:w_k}) to complete the proof of this lemma.

From Lemma \ref{lemma:2} together with $0 < \alpha_{\min} \le \alpha_{{\rm{s}},i}$ in (\ref{eqn:price}), we have $\mv{B}_{\mu,\nu} \succ \mv{0}$. Thus, it follows that $\tilde{\mv{V}}_{k} ^H\mv{B}_{\mu,\nu}\tilde{\mv{V}}_{k}$ is of full rank. Without loss of generality, we can set
\begin{align}\label{eqn:w_tilde}
\tilde{\mv{w}}_k = \left(\tilde{\mv{V}}_{k} ^H\mv{B}_{\mu,\nu}\tilde{\mv{V}}_{k}\right)^{-1/2}\tilde{\mv{w}}'_k.
\end{align}
Accordingly, problem (\ref{eqn:dual_function1:ZF}) becomes
\begin{align*}
\mathop{\mathtt{min}}\limits_{\tilde{\mv{w}}'_k}~&\|\tilde{\mv{w}}'_k\|^2\\
\mathtt{s.t.}~&\mv{h}_k^H\tilde{\mv{V}}_{k}\left(\tilde{\mv{V}}_{k} ^H\mv{B}_{\mu,\nu}\tilde{\mv{V}}_{k}\right)^{-1/2}\tilde{\mv{w}}'_k\ge \sigma_{k}\sqrt{\gamma_k}, \forall k\in\mathcal{K},
\end{align*}
for which the optimal solution is given by
\begin{align}
\tilde{\mv{w}}'_k = & \frac{\sigma_{k}\sqrt{\gamma_k}}{\left|\mv{h}_k^H\tilde{\mv{V}}_{k}\left(\tilde{\mv{V}}_{k} ^H\mv{B}_{\mu,\nu}\tilde{\mv{V}}_{k}\right)^{-1}\tilde{\mv{V}}_{k}^H\mv{h}_k\right|}  \nonumber\\ &\cdot \left(\tilde{\mv{V}}_{k} ^H\mv{B}_{\mu,\nu}\tilde{\mv{V}}_{k}\right)^{-1/2}\tilde{\mv{V}}_{k}^H\mv{h}_k,k\in\mathcal{K}. \label{eqn:solution:ZF2}
\end{align}
By combining (\ref{eqn:w_tilde}) and (\ref{eqn:solution:ZF2}), it follows that the optimal solution to (\ref{eqn:dual_function1:ZF}) is given in (\ref{eqn:w_k}). Therefore, Lemma \ref{lemma:4.3} is proved.

\end{document}